\documentclass[copyright,creativecommons,sharealike,leqno]{eptcs}
\providecommand{\event}{LINEARITY 2014}

\usepackage[utf8x]{inputenc}

\title{Geometry of Resource Interaction -- \textit{A Minimalist Approach}}
\author{Marco Solieri
\thanks{The author is deeply grateful to Michele Pagani and Stefano 
Guerrini for their advice.}\ \ 
\thanks{Partially supported by the ANR project ANR-2010-BLAN-021301 LOGOI.}
\institute{
    Université Paris 13, Sorbonne Paris Cité --
    Laboratoire d'Informatique de Paris-Nord, CNRS --
    Villetaneuse, France\\
    Università di Bologna --
    Dipartimento di Informatica, Scienza e Ingegneria, INRIA --
    Bologna, Italy}
  \email{ms@xt3.it}
}
\def\titlerunning{The Geometry of Resource Interaction -- A Minimalist Approach}
\def\authorrunning{M. Solieri}

\usepackage[utf8x]{inputenc}
\usepackage{url}
\usepackage{hyperref}
\hypersetup{
  colorlinks,
  citecolor=Navy,
  linkcolor=RawSienna,
  urlcolor=Blue,
  pdftitle={Geometry of Resource Interaction – A Minimalist Approach},
  pdfauthor={Marco Solieri},
  pdfkeywords={
  Geometry of Interaction,
  Resource lambda-calculus,
  Path,
  Dynamic Algebra,
  Interaction Nets,
  Semantics,
  Non-determinism,
  Linearity}
}

\usepackage{xspace}
\usepackage[usenames,dvipsnames,svgnames]{xcolor}

\usepackage{framed}
\usepackage{multirow}

\usepackage{graphicx}

\usepackage{caption}
\usepackage{subcaption}
\usepackage{float}
\floatstyle{ruled}
\restylefloat{figure}
\restylefloat{table}

\usepackage{tikz}
\usetikzlibrary{calc}
\usetikzlibrary{matrix}
\usetikzlibrary{positioning}
\usetikzlibrary{scopes}
\usetikzlibrary{matrix}
\usetikzlibrary{shapes.geometric}

\usepackage{stmaryrd}
\usepackage{amsmath,amsfonts,amssymb}
\usepackage{cmll}
\usepackage{mathtools}
\usepackage{textcomp}

\usepackage{amsthm}
\theoremstyle{plain}
    \newtheorem{theorem}{Theorem}
    \newtheorem{lemma}{Lemma}
    \newtheorem{corollary}{Corollary}

\theoremstyle{definition}
\newtheorem{definition}{Definition}
\newtheorem{example}{Example}
\theoremstyle{remark}

\usepackage{thmtools}

\usepackage{enumitem}

\renewcommand{\emph}[1]{\textsl{#1}}

\newcommand{\inpt}{\mathsf{in}}
\newcommand{\outpt}{\mathsf{out}}

\newcommand{\opconst}{\bigstar}
\newcommand{\const}{\star}
\newcommand{\wunit}{\star}

\newcommand{\pimp}{\multimap}
\newcommand{\nimp}{\negt{\multimap}}
\newcommand{\negt}[1]{\bar{#1}}

\newcommand{\lltens}{\otimes}
\newcommand{\llpar}{\parr}

\newcommand{\link}[3]{\langle#1 \; (\mathbin{#2}) \; #3\rangle}
\newcommand{\pimplink}[3]{\link {#1,#2} \pimp {#3}}
\newcommand{\nimplink}[3]{\link {#1,#2} \nimp {#3}}
\newcommand{\oclink}[2]{\link {#1}{\oc\,}{#2}}
\newcommand{\wnlink}[2]{\link {#1}{\wn}{#2}}

\newcommand{\constlink}[1]{\link {\!\!} \bigstar {#1}}

\newcommand{\graph}{\mathcal{G}}
\newcommand{\net}{\mathcal{N}}
\newcommand{\met}{\mathcal{M}}

\newcommand{\conc}{\!::\!}
\newcommand{\eqtd}{\!\equiv\!}

\newcommand{\weight}[1]{\mathfrak{w}(#1)}
\newcommand{\epaths}[1]{E(#1)}
\newcommand{\ecpaths}[1]{E^+(#1)}
\newcommand{\exec}[1]{\mathfrak{Ex}(#1)}

\newcommand{\reduce}{\rightarrow}

\newcommand{\nf}[1][\cdot]{\mathsf{NF}{(#1)}}
\newcommand{\simterms}{\Delta}
\newcommand{\simpolterms}{\Delta^!}
\newcommand{\nat}{\mathbf{N}}
\newcommand{\freemod}[2]{#1\langle #2 \rangle}
\newcommand{\terms}{\freemod{\nat}{\simterms}}
\newcommand{\polterms}{\freemod{\nat}{\simpolterms}}

\newcommand{\pretransl}[1]{\llparenthesis #1 \rrparenthesis}
\newcommand{\transl}[1]{\llbracket #1 \rrbracket}

\newcommand{\perm}{\sigma}
\newcommand{\permn}{\sigma_n}

\newcommand{\permns}{{S_n}}
\newcommand{\permms}{{S_m}}
\newcommand{\permgraph}{\sigma_\graph}
\newcommand{\permgraphs}{S_\graph}
\newcommand{\permnet}{\sigma_\net}
\newcommand{\permnets}{S_\net}
\newcommand{\permweight}[2]{\mathfrak{w}^{#1}(#2)}

\newcommand{\sumnet}{\mathcal{S}}
\newcommand{\rlstar}{{\mathfrak{rL}^*}}

\begin{document}
\maketitle


\begin{abstract}
The Resource $\lambda$-calculus is a variation of the $\lambda$-calculus 
where arguments can be superposed and must be linearly used.
Hence it is a model for linear and non-deterministic programming languages, 
and the target language of Taylor-Ehrhard expansion of $\lambda$-terms.
In a strictly typed restriction of the  Resource $\lambda$-calculus, we study 
the notion of path persistence,
and we define a Geometry of Interaction that characterises it,
is invariant under reduction,
and counts addends in normal forms.
\end{abstract}

\section*{Introduction}

\paragraph{Geometry of Interaction}
The dynamics of $\beta$-reduction or cut elimination can be described in a 
purely geometric way ---studying paths in the graphs that represent terms 
or proofs, and looking at those which are \textit{persistent}, i.e. that have a 
residual path in any reduct.
The quest for an effective semantical characterisation of persistence 
separately produced three notions of paths:
\textit{legality}, formulated by topological conditions about symmetries on 
some 
cycles \cite{AspertiLaneve:1995:TCS};
\textit{consistency}, expressed similarly to a token-machine execution 
\cite{GonthierAbadiLevy:1992:optimal} and developed to study the optimal 
reduction;
and \textit{regularity}, defined by a dynamic algebra
\cite{Girard:1989, DanosRegnier:1995}.
The notions are equivalent \cite{AspertiDanosLaneveRegnier:1994}, and their
common core idea ---describing computation by local and asynchronous 
conditions on routing of paths--- inspired the design of efficient parallel 
abstract machines 
\cite[for instance]{Mackie:1995, Pinto:2001}.
More recently, the Geometry of Interaction (GoI) approach has been fruitfully 
employed for semantical investigations which characterised quantitative 
properties of programs, mainly the complexity of their execution time 
\cite{DalLago:2009}.

\paragraph{Taylor-Ehrhard expansion of $\lambda$-terms and the Resource 
Calculus}
Linear Logic's decomposition of the intuitionistic implication unveiled the 
relation between the algebraic concept of linearity to the computational 
property of a function argument to be used exactly once.
Such a decomposition was then applied not only at the level of types, but also 
at the level of terms, in particular extending the $\lambda$-calculus with 
differential constructors and linear combinations of ordinary terms 
\cite{EhrhardRegnier:2003}.
These constructions allow to consider the complete Taylor expansion of a term, 
i.e. the infinite series of all the approximations of the reduction of a term, 
which was thus shown to commute with computation of Böhm trees.
The ideal target language for the expansion was isolated as the Resource 
$\lambda$-calculus (RC), which is a promotion-free restriction of the 
Differential $\lambda$-calculus \cite{EhrhardRegnier:2006}.
Taylor-Ehrhard expansion originated various investigations on quantitative 
semantics, using the concept of power series for describing program evaluation, 
and has been applied in various non-standard models of computation, e.g.
\cite[for instance]{DanosEhrhard:2011,PaganiSelingerValiron:2014}.

\paragraph{Aim and results}
How can the two aforementioned semantics approaches interact?
What is the relation between the GoI's execution formula and the expansion of 
$\beta$-reduction?
We present the first steps towards this direction.
After having concisely introduced RC (§\ref{sec:terms}),
we consider the Resource Interaction Nets (RINs),
that are the type-restricted translation of 
resource terms into Differential Interaction Nets (§\ref{sec:nets}).
We then study the appropriate notion of paths (§\ref{sec:paths}), extending the 
notion of persistence to paths in RINs dealing with the fact that 
the reduct of a term $t$ is a sum of nets $t_1+\ldots+t_n$.
In particular, we observe that every path of $t_i$ has to be a residual of some 
path in $t$, and that the reduction strongly normalises.
Thus, we say a path of $t$ to be persistent whenever it has a residual in at 
least one of the addends of the reduct of $t$.
Restricting the calculus to the constant type, whose only inhabitant is the 
value $\const$, we have $t \reduce \const + \ldots + \const$.
Now there is only one persistent path of $\const$, the trivial one, therefore 
we prove that persistent paths of $t$ are as many as persistent paths of its 
normal form (\autoref{thm:RIN:path-red:bij}).
Furthermore, we define a suitable GoI for RC, in order to characterise 
persistence (§\ref{sec:Ex}).
We define the notion of regularity by $\rlstar$, an appropriate variant of 
the Dynamic Algebra, where exponentials ($\oc$ and $\wn$) become a sort of 
$n$-ary multiplicatives (resp. $\lltens$ and $\llpar$), whose premises are 
not ordered.
Morally, they are the sum of the multiplicatives we obtain by considering all 
the $n!$ permutations of their premises.
We show our algebra is invariant under reduction 
(\autoref{thm:RIN:weight-inv}), 
from which we obtain the equivalence of persistence to regularity
(\autoref{thm:RIN:regular}) and also that the number of addends in a normal 
form is 
equal to the number of regular paths (\autoref{cor:RIN:path-red:number}).

\paragraph{Related works}
In a very closely related work by De Falco \cite{DeFalco:2008}, a GoI 
construction for DINs is formulated.
Besides the similarities in the technical setting of DINs, our formulation 
turns out to be simpler and more effective, mainly thanks to:
(1) the restriction to closed and ground-typed resource nets,
(2) the associative syntax we adopted for exponential links, and
(3) the stronger notion of path we use.
The first simplifies the shape of paths being persistent, because it implies 
that they are palindrome ---they go from the root to the $\const$ and back to 
the root--- and unique in every normal net/term.
The second simplifies the management of the exponential links, because 
it ensures associativity and delimits their dynamics in only one pair of links, 
while in \cite{DeFalco:2008} this property was completely lost and the system 
more verbose.
De Falco uses binary exponential links and introduces a syntactical embedding 
of the sum in nets by mean of binary links of named sums, and then recover 
associativity with an equivalence on nets.
Compared to ours, their choice results in a drastically more complex 
GoI construction, even though the paper hints at the extensibility with 
promotion (corresponding to the full Differential $\lambda$-calculus) or even 
additives.
The third ingredient allows us to consider full reduction, 
i.e. including the annihilating rule, while in \cite{DeFalco:2008} a 
``weak'' variant is studied, where this kind of redexes are frozen, and 
the GoI only characterises the corresponding notion of ``weak-persistence''.
Indeed, we restrict to paths that cross every exponential in the net (we 
prove it is always true, in case of persistence), thus whenever $t \reduce 0$ 
a path necessarily crosses the annihilating redex, and the dynamic algebra is 
able to detect it.

\section{Resource calculus}
\label{sec:terms}

The Resource Calculus is, on one hand, a linear and thus finitary restriction 
of the $\lambda$-calculus: an argument $[s]$ must be used by an application 
$t\ [s]$ exactly once, i.e. it cannot be duplicated nor erased, so every 
reduction enjoys strong normalisation.
On the other hand, it adds non-determinism to the $\lambda$-calculus, because 
the argument is a finite multiset of ordinary terms.
The reduct is then defined as the superposition of all the possible ways of 
substituting each of the arguments, i.e. a sum.
When arguments provided to a function are insufficient or excess the 
function's request, i.e. the number of variable occurrences, then the 
computation is deadlocked and the application reduces to $0$.
We shall omit the ``resource'' qualification in the terminology.

\begin{definition}[Syntax]\label{def:terms:syntax}
Let $\mathbb{V}$ be the grammar of a denumerable set of variable symbols
$x, y, z, \ldots$.
Then, the set $\simterms$ of the \textit{simple terms} and the set 
$\simpolterms$ of \textit{simple polyterms} are inductively and mutually 
generated by the following grammars.
\begin{align}
  \text{Simple terms: }
  \mathbb{M} \ ::= \ 
  \const \;|\;
  \mathbb{V} \;|\;
  \lambda \mathbb{V}.\mathbb{M} \;|\;
  \mathbb{M}\ \mathbb{B}
  &&
  &\text{Simple polyterms: }
  \mathbb{B} \ ::= \ 
    1 \;|\;
    [\mathbb{M}] \;|\;
    \mathbb{B} \cdot \mathbb{B}
\end{align}
Where: $\const$ is the constant dummy value, brackets delimit multisets, 
$\cdot$ is the multiset union (associative and commutative), $1$ is the empty 
multiset (neutral element of $\cdot$).
So that $([x]\cdot 1) \cdot [y] = [x,y]$.
Simple terms are denoted by the lowercase letters of the latin alphabet 
around $t$, polyterms in uppercase letters around $T$.
The set $\terms$ of \textit{terms} (resp. the set $\polterms$ of 
\textit{polyterms}) is the set of finite formal sums of simple terms
(resp. polyterms) over the semiring $\nat$ of natural numbers.
We also assume all syntactic constructors of simple terms and polyterms to be 
extended to sums by (bi-) linearity.
E.g.
$(\lambda x.(2x + y)) [z+4u]$ is a notational convention for
$2(\lambda x.x)[z] + 8(\lambda x.x)[u] + (\lambda x.y)[z]+ 4(\lambda x.y)[u]$.
\end{definition}

\begin{definition}[Reduction]\label{def:RT:reduction}
A redex is a simple term in the form: $(\lambda x. s) T$.
Let the free occurrences of $x$ in $s$ be $\lbrace x_1, \ldots, x_m\rbrace$.
The \textit{reduction} is the relation $\reduce$ between polyterms 
obtained by the context closure and the linear extension to sum of the 
following elementary reduction rule.
\begin{equation}
  \label{eq:RT:red}
    (\lambda x . s) \ [t_1,\ldots, t_n] \reduce
    \begin{dcases}
      \sum_{\permn \in \permns}
        s\ \lbrace
          t_1/x_{\permn(1)}, \ldots, t_n/x_{\permn(n)}
        \rbrace
        &\text{if } n=m\\
      0
        &\text{if } n \neq m
    \end{dcases}
\end{equation}
Where $\permns$ denotes the set of permutations of the first $n$ naturals,
and $\lbrace t/x \rbrace$ is the usual capture-avoiding substitution.
If $t \reduce^* t' \not \reduce$, where $\reduce^*$ is the reflexive transitive 
closure of $\reduce$, we write $\nf[t]=t'$.
\end{definition}

\begin{example}
\label{ex:terms}
Let
  $I = I' = \lambda x . x$
and also let
  $t = \lambda f. f_1 [ f_2 [\const] ]$.
Then
  $t [I,I'] \reduce
    f_1 [ f_2 [\const]] \lbrace I/f_1, I'/f_2 \rbrace + 
    f_1 [ f_2 [\const]] \lbrace I/f_2, I'/f_1 \rbrace$
that is $I [I' [\const]] + I' [I [\const]]$,
normalising to
  $I'[\const] + I [\const]
    \reduce 2\const$.
Note also a case of annihilation in $t[I] \reduce 0$.
Finally, observe that if $s = (\lambda x.\const) T \reduce \const$ then $T$ 
must be $1$ (otherwise $s \reduce 0$).
\end{example}

\section{Resource nets}
\label{sec:nets}

A resource net is a graphical representation of a typed term by means of a 
syntax borrowed from Linear Logic proof nets, where $n$-ary $\wn$ links have 
a symmetrical dual.
The exponential modality is however deprived of promotion, so that it merely 
represents superposition of proofs and contexts.

\subsection{Pre-nets}
\label{sec:nets:prenets}

\begin{definition}[Links] \label{def:link}
Given a denumerable set of symbols called \textit{vertices}, a \textit{link} 
is a triple $(P, K, C)$, where:
$P$ is a sequence of vertices, called premises;
$K$ is a kind, i.e. an element in the set
  $\lbrace \opconst, \pimp, \nimp, \oc, \wn \rbrace$;
$C$ is a singleton of a vertex, called conclusion, disjoint from $P$.
A link $l = ((u_1, \ldots, u_n), \kappa, \lbrace v \rbrace)$ will be denoted 
as $\link{u_1, \ldots, u_n}{\kappa}{v}$, or depicted as in 
\autoref{fig:net:links}.
The \textit{polarity} of a vertex associated by a link is an element in 
$\lbrace \inpt,\outpt \rbrace$
and we say they are opposite,
and the \textit{arity} of a link is the length of its premises' sequence;
both are determined by the link's kind,
as shown in \autoref{fig:net:links}.
When $v \in P(l) \cup C(l)$ for some vertex $v$ and link $l$, we write that 
$v$ is linked by $l$, or $v \in l$.
The exponential links $\oc$ and $\wn$ whose arity is $0$ are 
respectively called co-weakening and weakening.
In the graphical representations, vertices of a link shall be placed following 
the usual convention for graphs of $\lambda$-calculus ($\outpt$s on the top, 
and $\inpt$s on the bottom); the arrow line shall be used to distinguish 
the conclusion of a link.
\end{definition}

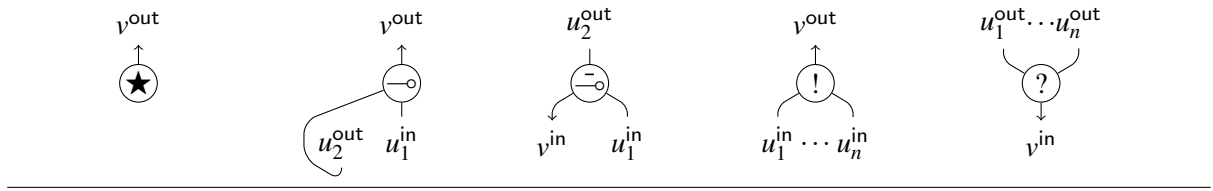
\begin{figure}
  \caption{Links: kind, arity and polarity associated to vertices.}
  \label{fig:net:links}
  \centering
  \begin{tikzpicture}[
    auto, scale=1,
    link/.style = { circle, minimum size=5mm, draw },
    ]
  { [xshift=0cm,yshift=1cm]
    \node[link] (star) at (1,0) {}; \node at (star) {$\bigstar$};
    \node (v) at (1,0.8) {$v^\outpt$};
    \draw[->] (star) -- (v);
  }{ [xshift=3.5cm,yshift=1cm]
    \node[link] (lambda) at (1,0) {}; \node at (lambda) {$\pimp$};
    \node (v) at (1,0.8) {$v^\outpt$};
    \node (u1) at (1,-0.8) {$u_1^\inpt$};
    \node (u2) at (.2,-0.8) {$u_2^\outpt$};
    \draw[->] (lambda) -- (v);
    \draw (lambda) -- (u1);
    \draw [rounded corners=5pt]
      (lambda)
      -- ++(-1.3,-0.5)
      -- ++(0,-0.5)
      -- ++(.5,-.3)
      -- (u2);
  }{ [xshift=6cm,yshift=1cm]
    \node[link] (app) at (1,0) {}; \node at (app) {$\nimp$};
    \node (u2) at (1,0.8) {$u_2^\outpt$};
    \node (u1) at (1.5,-0.8) {$u_1^\inpt$};
    \node (v)  at (.5,-0.8) {$v^\inpt$};
    \draw (app) -- (u2);
    \draw [rounded corners=3pt]
      (app)[->] -- ++(-.5,-.3) -- (v);
    \draw [rounded corners=3pt]
      (app) -- ++(.5,-.3) -- (u1);
  }{ [xshift=9cm,yshift=1cm]
      \node[link] (oc) at (1,0) {}; \node at (oc) {$\oc$};
      \node (v) at (1,0.8) {$v^\outpt$};
      \node (u1)  at (.5,-0.8) {$u_1^\inpt$};
      \node (un) at (1.5,-0.8) {$u_n^\inpt$};
      \node at (1,-0.8) {$\ldots$};
      \draw[->] (oc) -- (v);
      \draw [rounded corners=3pt]
        (oc) -- ++(-.5,-.3) -- (u1);
      \draw [rounded corners=3pt]
        (oc) -- ++(.5,-.3) -- (un);
  }{ [xshift=12cm,yshift=1cm]
      \node[link] (wn) at (1,0) {}; \node at (wn) {$\wn$};
      \node (u1)  at (.5,0.8) {$u_1^\outpt$};
      \node (un) at (1.5,0.8) {$u_n^\outpt$};
      \node (v) at (1,-0.8) {$v^\inpt$};
      \node at (1,0.8) {$\ldots$};
      \draw [rounded corners=3pt]
        (wn) -- ++(-.5,.3) -- (u1);
      \draw [rounded corners=3pt]
        (wn) -- ++(.5,.3) -- (un);
      \draw[->] (wn) -- (v);
    }
  \end{tikzpicture}
\end{figure}

\begin{definition}[Types]\label{def:types}
A \textit{type}, or formula, is a word of the grammar given by
  $\mathbb{T} ::=
  \const \ |\ \mathbb{E} \pimp \mathbb{T}$
and
  $\mathbb{E} ::= \oc \mathbb{T}$,
where $\const$ is the only \textit{ground} type.
A \textit{typing function} $\mathcal{T}$ is a map from vertices to types such 
that, if $A,B$ are types, then $\mathcal{T}$ respects the following constraints.
Constant:
  $\constlink{\const}$.
Linear implications:
  $\pimplink{A}{B} {A\!\!\pimp\!\!B}$ and $\nimplink{A}{B}{A\!\!\pimp\!\!B}$.
Exponentials:
  $\oclink{A, \ldots, A}{\oc A}$ and $\wnlink{A, \ldots, A}{\oc A}$.
\end{definition}

\begin{definition}[Pre-nets] \label{def:prenet}
A \textit{simple pre-net} $\graph$ is a triple $(V, L, \mathcal{T})$,
where $V$ is a set of vertices, $L$ is a set of links and $\mathcal{T}$ a 
typing function on $V$, such that for every vertex $v \in V$ the followings 
holds:
\begin{enumerate}
\item
  there are at least one and at most two links $l, l'$ such that
  $l \ni v \in l'$, and when there is only $l$, then $v$ is called a 
  \textit{conclusion} of $\graph$;
\item 
  the set $C(\graph)$ of conclusions is non empty and when it is the singleton 
  $v$, then $\graph$ is called \textit{closed} and $v$ must be $\outpt$;
\item
  if $l \ni v \in l'$, then $l,l'$ associate opposite polarities to $v$.
\end{enumerate}
We shall write $V(\graph)$ to denote the set $V\in\graph$.
The type of a pre-net $\graph$ is the type $T$ associated to its only $\outpt$ 
conclusion, so we write $\graph:T$.
The \textit{interface} of a simple pre-net $\graph$ is the set $I(\graph)$ of 
all ordered pairs $(T,\mathsf{p})$ such that for all $v \in C (\graph)$, $v$ 
is of type $T$ and has polarity $\mathsf{p}$.
A general \textit{pre-net} is a linear combination of simple pre-nets
$\graph_1 + \ldots + \graph_n$, where for any $1 \leq i,j \leq n$, we have:
$V(\graph_i) \cap V(\graph_j)=\emptyset$ and 
$I(\graph_i)=I(\graph_{i+1})$.
We shall simply use $0$ to denote each of the empty sums of pre-nets having 
the same interface $I$, for every interface $I$.
\end{definition}

\subsection{Term translation and net reduction}
\label{sec:nets:transl}

As the usual translation of the $\lambda$-calculus into MELL proof nets, 
the $\pimp$-link is used for translating $\lambda$-abstraction, the 
$\nimp$-link for application, and the $\wn$-link for contracting together all 
the occurrences of the same variable.
In addition, we use $\oc$-link for polyterm and formal sum of nets for\ldots 
formal sum of terms.

\begin{definition}[Term translation]\label{def:RIN:transl}
Given a simple term $t$, the \textit{translation} $\transl{t}$ is a pre-net 
having one $\outpt$ conclusion and a possibly empty set of $\inpt$ conclusions.
The translation is defined in \autoref{fig:net:transl} where: the final step 
only adds a $\wn$-link on every occurrence of a free variable $x$, for all free 
variables of $t$; and the actual work is performed by the $\pretransl{t}$,
by induction on the syntax of $t$.
Moreover a sum of simple terms is translated to the sum of their translation, 
i.e.:
$\transl{t_1 + \ldots + t_k} = \transl{t_1}+\ldots+\transl{t_n}$.
\end{definition}

\begin{figure}
  \caption{Pre-translation and translation of simple terms into simple nets.}
  \label{fig:net:transl}
  \centering
  \begin{tikzpicture}[auto, scale=1,
    subnet/.style = {
      rectangle,
      rounded corners=5pt,
      very thick,
      draw
    },
    link/.style = { circle, minimum size=5mm, draw },
    contr/.style = {
      isosceles triangle,
      isosceles triangle apex angle=60,
      draw
  }]
  { [xshift=0.4cm,yshift=5.1cm]
    \node at (0.1,0.3) {$\pretransl{\const} =$};
    \node[link] (star) at (1,0) {}; \node at (star) {$\bigstar$};
    \node (v) at (1,0.7) {$v$};
    \draw [->] (star) -- (v);
  }{ [xshift=0.5cm,yshift=4cm]
    \node at (0.9,0) {$\pretransl{\lambda x.t} =$};
    \node[subnet] (sub) at (3,0) {$\ \quad\pretransl{t}\quad\ $};
    \node[link] (lambda) at (3,1) {}; \node at (lambda) {$\pimp$};
    \node (v) at (3,1.8) {$v$};
    \node[link] (contr) at (3,-1.4) {}; \node at (contr) {$\wn$};
    \draw[->] (lambda) -- (v);
    \draw (lambda) -- node {$u_1$} (sub);
    \draw[<-, rounded corners=10pt]
      (lambda) -- ++(-1.2,-0.7) -- ++(0,-1) --
      (3,-2) node {$u_2$} -- (contr);
    \draw[rounded corners=10pt]
      (sub.west) -- ++(0,-0.5) -- (contr);
    \node at (3,-0.6) {$\!\!w_1 \ldots w_n$};
    \draw[rounded corners=10pt]
      (sub.east) -- ++(0,-0.5) -- (contr);
  }{ [xshift=4.7cm,yshift=3.8cm]
    \node at (1,1) {$\pretransl{[s_1, \ldots s_n]} =$};
    \node[subnet] (sub1) at (1.9,0) {$\quad\pretransl{s_1}\quad$};
    \node at (3,0) {$\ldots$};
    \node[subnet] (subn) at (4.1,0) {$\quad\pretransl{s_n}\quad$};
    \node[link] (cocontr) at (3,1.1) {}; \node at (cocontr) {$\oc$};
    \node (v) at (3,2) {$v$};
    \node at (3,0.6) {$u_1 \ldots u_n$};
    \draw [<-] (v) -- (cocontr);
    \draw[rounded corners=5pt] (sub1) -- ++(0,0.5) -- (cocontr);
    \draw[rounded corners=5pt] (subn) -- ++(0,0.5) -- (cocontr);
  }{ [xshift=8.2cm,yshift=2.3cm]
    \node at (0.8,0.4) {$\pretransl{t S} =$};
    \node[subnet] (subt) at (2.2,0) {$\quad\pretransl{t}\quad$};
    \node[subnet] (subS) at (3.8,0) {$\quad\pretransl{S}\quad$};
    \node[link] (nimp) at (3,1) {}; \node at (nimp) {$\nimp$};
    \node (v) at (3,1.8) {$v$};
    \node at (3,0.6) {$u \quad w$};
    \draw (v) -- (nimp);
    \draw[<-, rounded corners=4pt] (subt) -- ++(0,0.5) -- (nimp);
    \draw[rounded corners=4pt] (subS) -- ++(0,0.5) -- (nimp);
  }{ [xshift=11.3cm,yshift=5.6cm]
    \node at (1,-0.9) {$\transl{t} =$};
    \node[subnet] (sub) at (3,0) {$\qquad\pretransl{t}\qquad$};
    \node[link] (contr1) at (2.5,-1.4) {}; \node at (contr1) {$\wn$};
    \node[link] (contrl) at (3.5,-1.4) {}; \node at (contrl) {$\wn$};
    \node (v1) at (2.5,-2.3) {$v_1$};
    \node (vl) at (3.5,-2.3) {$v_l$};
    \draw[->] (contr1) -- (v1);
    \draw[->] (contrl) -- (vl);
    \draw[rounded corners=15pt] (sub.west) -- ++(0,-0.5) -- (contr1);
    \draw[rounded corners=15pt] (sub.south) -- ++(0,-0.5) -- (contr1);
    \draw[rounded corners=15pt] (sub.south) -- ++(0,-0.5) -- (contrl);
    \draw[rounded corners=15pt] (sub.east) -- ++(0,-0.5) -- (contrl);
    \node at (2.5,-.9) {$\ldots$};
    \node at (3.5,-.9) {$\ldots$};
    \node at (3,-2.3) {$\ldots$};
    \node at (1.8,-.6) {$v_{1_1}$};
    \node at (2.7,-.6) {$v_{1_j}$};
    \node at (3.3,-.6) {$v_{l_1}$};
    \node at (4.3,-.6) {$v_{l_k}$};
  }{ [xshift=0.8cm,yshift=2.2cm]
    \node at (0,0) {$\pretransl{x} =\ \ v$};
  }
  \end{tikzpicture}
\end{figure}
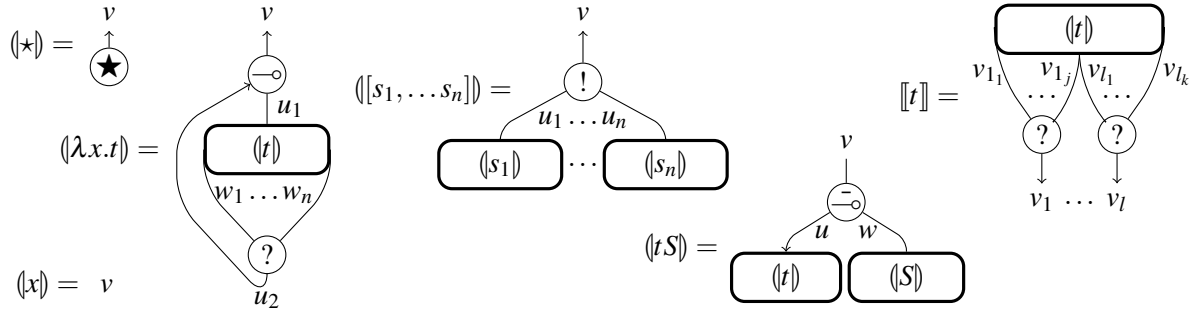

Note that a net translation is always defined for simple terms while it is 
not for general terms, because of possible incompatibility in the interfaces of 
translated addends.

\begin{definition}[Resource permutations]\label{def:RIN:permnet}
Given a simple pre-net $\graph$, a \textit{resource permutation} $\permgraph$ 
is a total function from the set of $\oc$-links in $\net$ to $\bigcup_n 
\permns$ such that
if a link $l$ has arity $m$,
  then $\permnet(l)$ is an element $\sigma_m$ of $\permms$.
We shall also write $\perm_l$ for $\permgraph(l)$ and denote the set of 
resource permutation of $\graph$ as $\permgraphs$.
\end{definition}

\begin{definition}[Resource net reduction]\label{def:nets:red}
The \textit{redex} of a cut $w$ in a simple pre-net is the pair of links having 
$w$ as conclusion.
The \textit{simple reduction} $\reduce$ is the graph-rewrit\-ing relation from 
simple pre-nets to pre-nets defined by the following elementary reduction steps,
also depicted in \autoref{fig:nets:red}, omitting contexts.
\begin{align}
\label{eq:RIN:red:impl}
  \graph,\ \pimplink{u}{v}{w},\ \nimplink{u'}{v'}{w}
  \ \ &\reduce\ \ 
  \graph [v \eqtd v',u \eqtd u']
\\
\label{eq:RIN:red:exp}
  \graph,\ \oclink{v_1,\ldots,v_n}{w},\ \wnlink{u_1,\ldots, u_m}{w}
  \ \ &\reduce\ \ 
  \begin{dcases}
    \sum_{\permn \in \permns}\ \graph_{\permn}
      [v_1 \eqtd u_{\permn(1)}, \ldots, v_n \eqtd u_{\permn(n)}]
      &\text{if } n=m\\
    0
      &\text{if } n \neq m
  \end{dcases}
\end{align}
Where $\graph_i [v \eqtd u]$ denotes the $i$-th copy of the pre-net $\graph$,
where the vertices $v,u$ have been equated.
In such a case, we say then there is a \textit{simple reduction step} $\rho: 
\graph \reduce \sumnet$, where $\sumnet$ is a sum of simple pre-nets and is 
also written as $\rho(\graph)$.
The \textit{reduction} is the extension of the simple reduction to formal sums 
of simple pre-nets: if
  $\graph \reduce \sumnet$,
then
  $\graph+\sumnet' \reduce \sumnet + \sumnet'$.
If
  $\graph \reduce^* \sumnet' \not \reduce$,
we write $\nf[\graph]=\sumnet'$.
\end{definition}

\begin{figure}
  \caption{Cut elimination rules: linear implication, and exponential.}
  \label{fig:nets:red}
  \centering
  \begin{tikzpicture}[
    auto, scale=1,
    link/.style = { circle, minimum size=5mm, draw },
    subnet/.style = { rectangle, rounded corners=5pt, very thick, draw}]
{ [xshift=-2cm,yshift=.7cm]
    \node[link] (lambda) at (0,0) {}; \node at (lambda) {$\pimp$};
    \node[link] (app) at (.7,1) {}; \node at (app) {$\nimp$};
    \node (vroot)  at (.7,1.8) {$v'$};
    \node (varg) at (1.4,.2) {$u'$};
    \node (vbody)  at (0,-.8) {$v$};
    \node (vvar)  at (-.4,-1) {$u$};
    \draw [rounded corners=5pt] (app) -- ++(-.7,-.3) -- node {$w$} (lambda);
    \draw [<-, rounded corners=5pt] (app) -- ++(.7,-.3) -- (varg);
    \draw [->] (vroot) -- (app);
    \draw [<-] (lambda) -- (vbody);
    \draw [<-, rounded corners=5pt]
      (lambda) -- ++(-.8,-.8) -- ++ (0,-.6) -- ++ (.4,0) -- (vvar);
    \node at (2.3,.4) {\Large $\reduce$};
    \node at (3.5,0.7) {$u \eqtd u'$};
    \node at (3,-.7) {$v \eqtd v'$};
}{ [xshift=5.5cm,yshift=0cm]
    \node[link] (ccontr) at (1,1.5) {}; \node at (ccontr) {$!$};
    \node[link] (contr)   at (1,.5) {}; \node at (contr) {$?$};
    \node (v1)  at (.3,2.2) {$v_1$};
    \node at (1,2.2) {$\ldots$};
    \node (vn) at (1.7,2.2) {$v_n$};
    \node (u1)  at (.3,-.2) {$u_1$};
    \node at (1,-.2) {$\ldots$};
    \node (um) at (1.7,-.2) {$u_m$};
    \draw (contr) -- node {$w$} (ccontr);
    \draw [rounded corners=5pt] (v1) -- ++(0,-.5) -- (ccontr);
    \draw [rounded corners=5pt] (vn) -- ++(0,-.5) -- (ccontr);
    \draw [rounded corners=5pt] (u1) -- ++(0,.5) -- (contr);
    \draw [rounded corners=5pt] (um) -- ++(0,.5) -- (contr);
    \node at (4,0.7) {
      {\Large $\xrightarrow{n=m}
        \displaystyle\sum_{\permn\in\permns}\;$}
      $v_i \eqtd v_{\permn(i)}$
    };
  \node at (-1,1) {
    {\Large $0 \; \xleftarrow{n\neq m}$} };}
  \end{tikzpicture}
\end{figure}
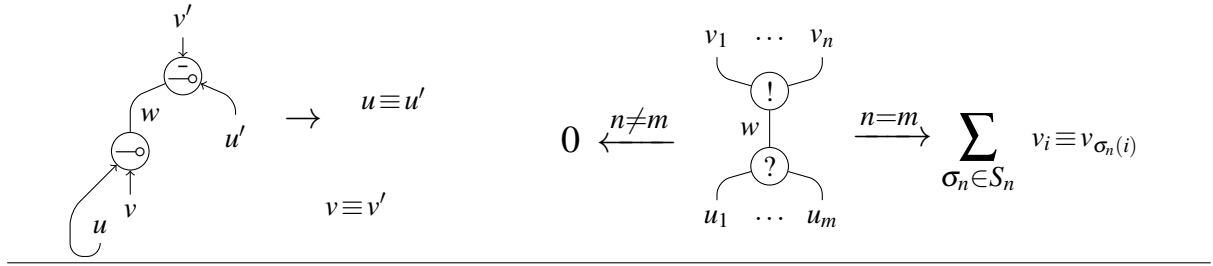

\begin{definition}[Resource interaction nets]\label{def:RIN}
Let $t\in \simterms$ and $\transl{t} \reduce^* \sumnet$, for a sum
$\sumnet = \net_1 + \ldots + \net_n$,
where each $\net_i$ is a pre-net.
Then $\net_i$ is called a \textit{simple resource interaction net}
and $\sumnet$ a \textit{resource interaction net}.
From now on we shall again avoid to repeat the ``resource interaction'' naming 
of nets.
\end{definition}

We recall that the net reduction can simulate the term reduction and strongly 
normalises.

\begin{example}\label{ex:nets}
Consider $\delta = \lambda x. x [x]$
and notice $\transl{\delta}$ is not a pre-net,
because a typing function on the structure of vertices and links does not exist.
Recall the terms $I$ and $t$ from \autoref{ex:terms}
and look at \autoref{fig:ex}.
On the left extremity:
  $\transl{I}$ is closed and $\transl{I}:\oc \const \pimp \const$.
On the middle left: $\net : \star$ is not a translation of a term,
but it is a net,
  because $\transl{t[x,y]} \reduce \net$ by eliminating a linear implication 
cut.
Also, $\net$ is not a closed net,
 because it has three conclusions: $v_1, z_1, z_2$.
On the right side: an exponential reduction step involving index permutation,
that rewrite $\net$ into a sum of two normal simple nets.
Observe the reduct is equal to
  $\transl{\ x[y[\star]]\ +\ y[x[\star]]\ }$.
Consider $\transl{\lambda f. f_1 [ f_2 [\const]][I,I] }$,
that is a closed net of type $\star$,
and observe the reduct $\met$ of the only linear implication cut
that is depicted in \autoref{fig:ex2}.
The normalisation requires: one exponential step (on the left),
two linear implication steps (on the right),
and finally four exponential steps (not showed) to reach the net
$
  \constlink{v_1 \eqtd v_8} +
  \constlink{v'_1 \eqtd v_8}
  =
 \transl{\const+\const}$.
\end{example}

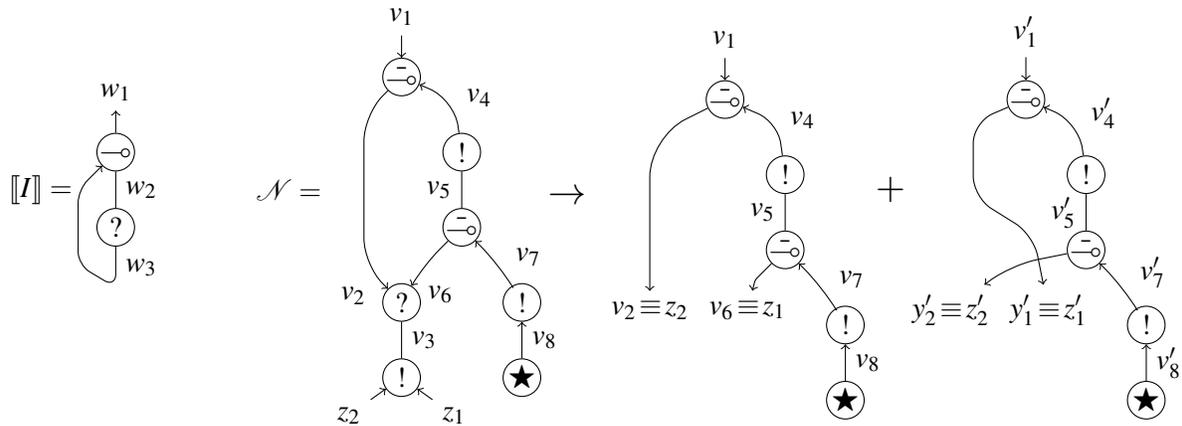
\begin{figure}[b]
  \centering
  \caption{Example: nets and reduction.}
  \label{fig:ex}
  \begin{tikzpicture}[
    auto, scale=1,
    link/.style = { circle, minimum size=5mm, draw },
    subnet/.style = { rectangle, rounded corners=5pt, very thick, draw} ]
  \node at (4.2,2.5) {$\net =$};
  { [xshift=-1.1cm,yshift=2cm]
    \node at (2,0.5) {$\transl{I} =$};
    \node[link] (lambda) at (3,1) {}; \node at (lambda) {$\pimp$};
    \node (v) at (3,1.8) {$w_1$};
    \node[link] (contr) at (3,0) {}; \node at (contr) {$\wn$};
    \draw[->] (lambda) -- (v);
    \draw (lambda) -- node {$w_2$} (contr);
    \draw[<-, rounded corners=5pt]
      (lambda) -- ++(-.5,-.5) -- ++(0,-1) --
      (3,-.8) -- node[swap] {$w_3$} (contr);
  }{ [xshift=4.7cm,yshift=1cm]
    \node (v1) at (1,3.8) {$v_1$};
    \node[link] (app1)    at (1,3) {}; \node at (app1) {$\nimp$};
    \node[link] (ccontr1) at (1.8,2) {}; \node at (ccontr1) {$\oc$};
    \node[link] (app2)    at (1.8,1) {}; \node at (app2) {$\nimp$};
    \node[link] (contr)   at (1,0) {}; \node at (contr) {$\wn$};
    \node[link] (ccontri) at (1,-1) {}; \node at (ccontri) {$\oc$};
    \node[link] (ccontr)  at (2.6,0) {}; \node at (ccontr) {$\oc$};
    \node[link] (star)    at (2.6,-1) {}; \node at (star) {$\bigstar$};
    \node (u1)  at (.3,-1.5) {$z_2$};
    \node (u1') at (1.7,-1.5) {$z_1$};
    \draw [<-] (app1) -- (v1);
    \draw[->, rounded corners=10pt]
      (app1) -- ++(-.5,-.5) -- ++(0,-2) -- node[swap] {$v_2$}  (contr);
    \draw[<-, rounded corners=10pt]
      (app1) -- ++(.7,-.3) -- node {$v_4$} (ccontr1);
    \draw (app2) -- node {$v_5$} (ccontr1);
    \draw[->, rounded corners=8pt]
      (app2) -- ++(-.5,-.5) -- node {$v_6$} (contr);
    \draw[<-, rounded corners=10pt]
      (app2) -- ++(.5,-.5) -- node {$v_7$} (ccontr);
    \draw (contr) -- node {$v_3$}(ccontri);
    \draw [<-] (ccontr) -- node {$v_8$}(star);
    \draw [<-](ccontri) -- (u1);
    \draw [<-](ccontri) -- (u1');
  }
  \node at (7.9,2.4) {\Large $\reduce$};
  { [xshift=11cm,yshift=0.7cm]
    { [xshift=-2cm,yshift=0cm]
    \node (v1) at (1,3.8) {$v_1$};
    \node[link] (app1)    at (1,3) {}; \node at (app1) {$\nimp$};
    \node[link] (ccontr1) at (1.8,2) {}; \node at (ccontr1) {$\oc$};
    \node[link] (app2)    at (1.8,1) {}; \node at (app2) {$\nimp$};
    \node[link] (ccontr)  at (2.6,0) {}; \node at (ccontr) {$\oc$};
    \node[link] (star)    at (2.6,-1) {}; \node at (star) {$\bigstar$};
    \node (u1)  at (0,0.2) {$v_2 \eqtd z_2$};
    \node (u1') at (1.3,0.2) {$v_6 \eqtd z_1$};
    \draw [<-] (app1) -- (v1);
    \draw[->, rounded corners=15pt] (app1) -- ++(-1,-.5) -- (u1);
    \draw[<-, rounded corners=8pt]
      (app1) -- ++(.7,-.3) -- node {$v_4$} (ccontr1);
    \draw (app2) -- node {$v_5$} (ccontr1);
    \draw [->, rounded corners=1pt] (app2) -- ++(-.4,-.4) -- (u1');
    \draw[<-, rounded corners=10pt] (app2) -- ++(.5,-.5) -- node {$v_7$} 
(ccontr);
    \draw [<-] (ccontr) -- node {$v_8$}(star);
    }
    \node at (1.2,1.8) {\Large$+$};
    { [xshift=2cm,yshift=0cm]
    \node (v1) at (1,3.9) {$v'_1$};
    \node[link] (app1)    at (1,3) {}; \node at (app1) {$\nimp$};
    \node[link] (ccontr1) at (1.8,2) {}; \node at (ccontr1) {$\oc$};
    \node[link] (app2)    at (1.8,1) {}; \node at (app2) {$\nimp$};
    \node[link] (ccontr)  at (2.6,0) {}; \node at (ccontr) {$\oc$};
    \node[link] (star)    at (2.6,-1) {}; \node at (star) {$\bigstar$};
    \node (u1')  at (0,0.2) {$y'_2 \eqtd z'_2$};
    \node (u1) at (1.3,0.2) {$y'_1 \eqtd z'_1$};
    \draw [<-] (app1) -- (v1);
    \draw[->, rounded corners=8pt]
      (app1) -- ++(-.7,-.3) -- ++(0,-1) -- ++(.8,-.5) -- (u1);
    \draw[<-, rounded corners=10pt]
      (app1) -- ++(.7,-.3) -- node {$v'_4$} (ccontr1);
    \draw (app2) -- node {$v'_5$} (ccontr1);
    \draw [->, rounded corners=10pt] (app2) -- ++(-1,-.2) -- (u1');
    \draw[<-, rounded corners=10pt] (app2) -- ++(.5,-.5) -- node {$v'_7$} 
(ccontr);
    \draw [<-] (ccontr) -- node {$v'_8$}(star);
    }
  }
  \end{tikzpicture}
\end{figure}

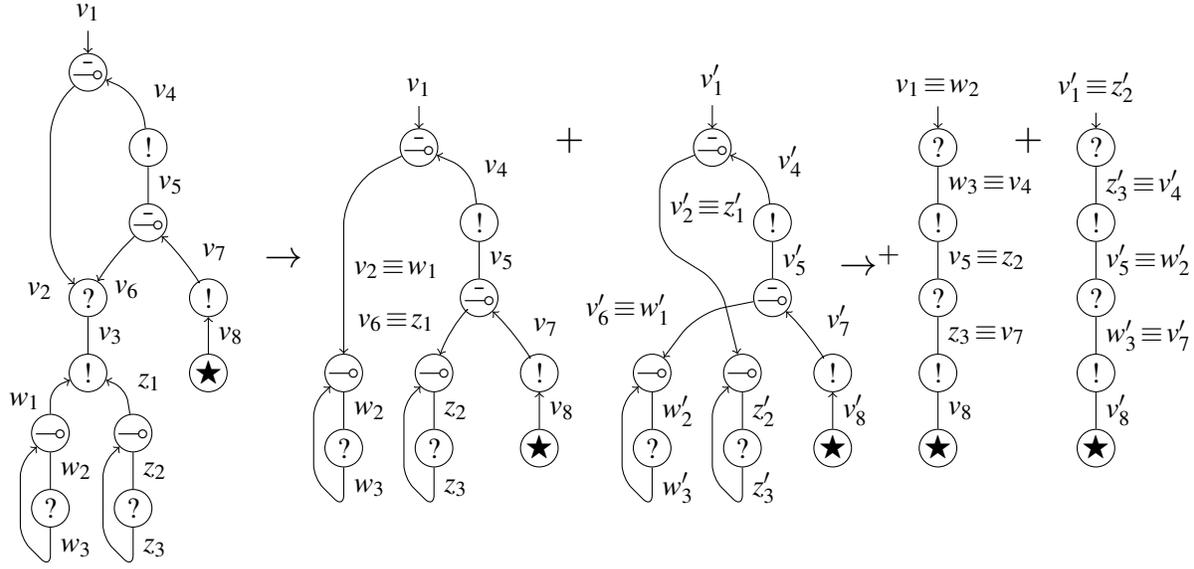
\begin{figure}
  \centering
  \caption{Example: nets reduction. Rightmost reduction is made of four steps, 
two on each addend.}
  \label{fig:ex2}
  \begin{tikzpicture}[
    auto, scale=1,
    link/.style = { circle, minimum size=5mm, draw },
    subnet/.style = { rectangle, rounded corners=5pt, very thick, draw} ]
  { [xshift=2.2cm,yshift=-1.8cm]
    \node[link] (lambda) at (3,1) {}; \node at (lambda) {$\pimp$};
    \node[link] (contr) at (3,0) {}; \node at (contr) {$\wn$};
    \draw (lambda) -- node {$w_2$} (contr);
    \draw[<-, rounded corners=5pt]
      (lambda) -- ++(-.4,-.4) -- ++(0,-1) --
      (3,-.8) -- node[swap] {$w_3$} (contr);
  }
  { [xshift=3.3cm,yshift=-1.8cm]
    \node[link] (lambda1) at (3,1) {}; \node at (lambda1) {$\pimp$};
    \node[link] (contr1) at (3,0) {}; \node at (contr1) {$\wn$};
    \draw (lambda1) -- node {$z_2$} (contr1);
    \draw[<-, rounded corners=5pt]
      (lambda1) -- ++(-.4,-.4) -- ++(0,-1) --
      (3,-.8) -- node[swap] {$z_3$} (contr1);
  }
  { [xshift=4.7cm,yshift=1cm]
    \node (v1) at (1,3.8) {$v_1$};
    \node[link] (app1)    at (1,3) {}; \node at (app1) {$\nimp$};
    \node[link] (ccontr1) at (1.8,2) {}; \node at (ccontr1) {$\oc$};
    \node[link] (app2)    at (1.8,1) {}; \node at (app2) {$\nimp$};
    \node[link] (contr)   at (1,0) {}; \node at (contr) {$\wn$};
    \node[link] (ccontri) at (1,-1) {}; \node at (ccontri) {$\oc$};
    \node[link] (ccontr)  at (2.6,0) {}; \node at (ccontr) {$\oc$};
    \node[link] (star)    at (2.6,-1) {}; \node at (star) {$\bigstar$};
    \draw [<-] (app1) -- (v1);
    \draw[->, rounded corners=10pt]
      (app1) -- ++(-.5,-.5) -- ++(0,-2) -- node[swap] {$v_2$}  (contr);
    \draw[<-, rounded corners=10pt]
      (app1) -- ++(.7,-.3) -- node {$v_4$} (ccontr1);
    \draw (app2) -- node[swap] {$v_5$} (ccontr1);
    \draw[->, rounded corners=8pt]
      (app2) -- ++(-.5,-.5) -- node {$v_6$} (contr);
    \draw[<-, rounded corners=10pt]
      (app2) -- ++(.5,-.5) -- node {$v_7$} (ccontr);
    \draw (contr) -- node {$v_3$}(ccontri);
    \draw [<-] (ccontr) -- node {$v_8$}(star);
    \draw [<-, rounded corners=5pt]
      (ccontri) -- ++(-.5,-.2) -- node[swap] {$w_1$} (lambda);
    \draw [<-, rounded corners=5pt]
      (ccontri) -- ++(.5,-.2) -- node[] {$z_1$}  (lambda1);
  }
  \node at (8.3,1.5) {\Large $\reduce$};
  { [xshift=11.1cm,yshift=0cm]
    { [xshift=-5cm,yshift=-1cm]
      \node[link] (lambda) at (3,1) {}; \node at (lambda) {$\pimp$};
      \node[link] (contr) at (3,0) {}; \node at (contr) {$\wn$};
      \draw (lambda) -- node {$w_2$} (contr);
      \draw[<-, rounded corners=5pt]
        (lambda) -- ++(-.4,-.4) -- ++(0,-1) --
        (3,-.8) -- node[swap] {$w_3$} (contr);
    }
    { [xshift=-3.8cm,yshift=-1cm]
      \node[link] (lambda1) at (3,1) {}; \node at (lambda1) {$\pimp$};
      \node[link] (contr1) at (3,0) {}; \node at (contr1) {$\wn$};
      \draw (lambda1) -- node {$z_2$} (contr1);
      \draw[<-, rounded corners=5pt]
        (lambda1) -- ++(-.4,-.4) -- ++(0,-1) --
        (3,-.8) -- node[swap] {$z_3$} (contr1);
    }
    { [xshift=-2cm,yshift=0cm]
      \node (v1) at (1,3.8) {$v_1$};
      \node[link] (app1)    at (1,3) {}; \node at (app1) {$\nimp$};
      \node[link] (ccontr1) at (1.8,2) {}; \node at (ccontr1) {$\oc$};
      \node[link] (app2)    at (1.8,1) {}; \node at (app2) {$\nimp$};
      \node[link] (ccontr)  at (2.6,0) {}; \node at (ccontr) {$\oc$};
      \node[link] (star)    at (2.6,-1) {}; \node at (star) {$\bigstar$};
      \draw [<-] (app1) -- (v1);
      \draw[->, rounded corners=15pt]
        (app1) -- ++(-1,-.5) -- node {$v_2 \eqtd w_1$} (lambda);
      \draw[<-, rounded corners=8pt]
        (app1) -- ++(.7,-.3) -- node {$v_4$} (ccontr1);
      \draw (app2) -- node[swap] {$v_5$} (ccontr1);
      \draw [->, rounded corners=10pt]
        (app2) -- ++(-.4,-.4) -- node[swap] {$v_6 \eqtd z_1$} (lambda1);
      \draw[<-, rounded corners=10pt]
        (app2) -- ++(.5,-.5) -- node {$v_7$} (ccontr);
      \draw [<-] (ccontr) -- node {$v_8$}(star);
    }
    \node at (1,3.1) {\Large$+$};
    { [xshift=-.9cm,yshift=-1cm]
      \node[link] (lambda) at (3,1) {}; \node at (lambda) {$\pimp$};
      \node[link] (contr) at (3,0) {}; \node at (contr) {$\wn$};
      \draw (lambda) -- node {$w'_2$} (contr);
      \draw[<-, rounded corners=5pt]
        (lambda) -- ++(-.4,-.4) -- ++(0,-1) --
        (3,-.8) -- node[swap] {$w'_3$} (contr);
    }
    { [xshift=0.3cm,yshift=-1cm]
      \node[link] (lambda1) at (3,1) {}; \node at (lambda1) {$\pimp$};
      \node[link] (contr1) at (3,0) {}; \node at (contr1) {$\wn$};
      \draw (lambda1) -- node {$z'_2$} (contr1);
      \draw[<-, rounded corners=5pt]
        (lambda1) -- ++(-.4,-.4) -- ++(0,-1) --
        (3,-.8) -- node[swap] {$z'_3$} (contr1);
    }
    { [xshift=1.9cm,yshift=0cm]
      \node (v1) at (1,3.9) {$v'_1$};
      \node[link] (app1)    at (1,3) {}; \node at (app1) {$\nimp$};
      \node[link] (ccontr1) at (1.8,2) {}; \node at (ccontr1) {$\oc$};
      \node[link] (app2)    at (1.8,1) {}; \node at (app2) {$\nimp$};
      \node[link] (ccontr)  at (2.6,0) {}; \node at (ccontr) {$\oc$};
      \node[link] (star)    at (2.6,-1) {}; \node at (star) {$\bigstar$};
      \draw [<-] (app1) -- (v1);
      \draw[->, rounded corners=8pt]
        (app1) -- ++(-.7,-.3) --  node {$v'_2 \eqtd z'_1$} ++(0,-1) --
        ++(.8,-.5) -- (lambda1);
      \draw[<-, rounded corners=10pt]
        (app1) -- ++(.7,-.3) -- node {$v'_4$} (ccontr1);
      \draw (app2) -- node[swap] {$v'_5$} (ccontr1);
      \draw [->, rounded corners=10pt]
        (app2) -- ++(-1,-.2) -- node[swap] {$v'_6 \eqtd w'_1$} (lambda);
      \draw[<-, rounded corners=10pt]
        (app2) -- ++(.5,-.5) -- node {$v'_7$} (ccontr);
      \draw [<-] (ccontr) -- node {$v'_8$}(star);
    }
  }
  \node at (16.1,1.5) {\Large $\rightarrow^+$};
  { [xshift=17cm,yshift=0cm]
    \node (v1) at (0,3.8) {$v_1 \eqtd w_2$};
    \node[link] (contr) at (0,3) {}; \node at (contr) {$\wn$};
    \node[link] (ccontr) at (0,2) {}; \node at (ccontr) {$\oc$};
    \node[link] (contr1) at (0,1) {}; \node at (contr1) {$\wn$};
    \node[link] (ccontr1) at (0,0) {}; \node at (ccontr1) {$\oc$};
    \node[link] (star)    at (0,-1) {}; \node at (star) {$\bigstar$};
    \draw [->] (v1) -- (contr);
    \draw (contr) -- node {$w_3 \eqtd v_4$} (ccontr);
    \draw (ccontr) -- node {$v_5 \eqtd z_2$} (contr1);
    \draw (contr1) -- node {$z_3 \eqtd v_7$} (ccontr1);
    \draw (ccontr1) -- node {$v_8$} (star);
  }
  \node at (18.2,3.1) {\Large $+$};
  { [xshift=19.1cm,yshift=0cm]
    \node (v1) at (0,3.8) {$v'_1 \eqtd z'_2$};
    \node[link] (contr) at (0,3) {}; \node at (contr) {$\wn$};
    \node[link] (ccontr) at (0,2) {}; \node at (ccontr) {$\oc$};
    \node[link] (contr1) at (0,1) {}; \node at (contr1) {$\wn$};
    \node[link] (ccontr1) at (0,0) {}; \node at (ccontr1) {$\oc$};
    \node[link] (star)    at (0,-1) {}; \node at (star) {$\bigstar$};
    \draw [->] (v1) -- (contr);
    \draw (contr) -- node {$z'_3 \eqtd v'_4$} (ccontr);
    \draw (ccontr) -- node {$v'_5 \eqtd w'_2$} (contr1);
    \draw (contr1) -- node {$w'_3 \eqtd v'_7$} (ccontr1);
    \draw (ccontr1) -- node {$v'_8$} (star);
  }
  \end{tikzpicture}
\end{figure}

\section{Paths}
\label{sec:paths}

\subsection{Definitions}
We introduce some basic definitions about the paths, where the most notable 
characterise the paths where the computation is visible (straightness) 
in its entirety (maximality and comprehensiveness).
This last notion is the only substantial difference with respect to the 
classic notion of path as formulated in \cite{DanosRegnier:1995}.
A superficial technical difference is the choice of using concatenation 
instead of composition as the basic relation on paths.

\begin{definition}[Path] \label{def:path}
  Given a simple net $\net$, two vertices $u, w \in \net$ are \textit{linked}, 
  or connected, if there is a link $l \in \net$ s.t.  $u,w \in l$.
  A \textit{path} $\pi = (v_1, \ldots, v_n)$ with $n>0$ in $\net$
  is a sequence of vertices s.t.
    for all $i < n$,
    the vertices $v_i, v_{i+1}$ are connected.
  We call $\pi$ \textit{trivial} if its lenght is $1$; 
  we say $\pi$ \textit{unitary} if is $2$, so that there is only one link
  crossed by $\pi$.
  \\
  Moreover, if $\pi$ crosses consecutively the same link $l$ more that once, 
  then $\pi$ is called \textit{bouncing}.
  If $l$ is not a $\const$-link and $\pi$ crosses $l$ through $v,v'$ such that
    $v,v' \in C(l)$ or $v,v' \in P(l)$,
  then $\pi$ is \textit{twisting}.
  When $\pi$ is both non-bouncing and non-twisting, $\pi$ is \textit{straight}.
  Moreover, $\pi$ is \textit{maximal} if there is no other path
  $\pi' \in \net$ s.t. $\pi \sqsubseteq \pi'$, where $\sqsubseteq$ is the 
  prefix order on sequences.
  Also, $\pi$ is \textit{comprehensive} when it crosses all the premises of all 
  the exponential links.
  Finally if $\pi$ is both straight and maximal, then $\pi$ is an 
  \textit{execution} path.
  In a net $\sumnet$, we denote with $\epaths{\sumnet}$
  (or with $\ecpaths{\sumnet}$) 
  the set of execution paths (respectively also comprehensive)
  in some simple $\net$ addend of $\sumnet$.
  \\
  Given two paths $\pi,\pi'$
  we denote the \textit{reversal} of $\pi$ as $\pi^-$,
  while the \textit{concatenation} of $\pi'$ to $\pi$ as $\pi \conc \pi'$.
\end{definition}

We can now concretely aim to define a proper notion of path persistence, that 
intuitively means ``having a residual'', so first we inspect and define the 
action of residual of path.
The case of linear implication is straightforward, because the rewriting is 
local and we only have to ensure that a path does not partially belong to 
a redex.
The case of exponential, instead, is rather more delicate, because the 
rewriting is global: a simple net rewrites to a sum of simple nets, hence a 
path may be duplicated in several addends or destroyed.
Which addends contain the residual(s) of a given crossing?
The net reduction consists of the sum of all the permutation of the indices of 
the $\oc$-links (cf. \autoref{def:nets:red}), thus each addend contains all 
and only the paths that respect the addend's own permutation, for any crossing 
of the redex.
If a path $\pi$ is persistent, then, there must be a permutation such that 
$\pi$ always crosses the redex respecting the correspondences fixed by the 
permutation.


\begin{definition}[Path residual]
\label{def:RIN:path-red}
Given a net $\net$ and a reduction $\rho$ on a redex $R \in \net$,
we say a path $\pi \in \net$ is \textit{long enough} for $R$ when neither its 
first nor its last vertex is the cut of $R$.
In such a case, we can express $\pi$ by isolating every \textit{crossing} of 
$R$, that is a maximal sub-sequence of $\pi$ entirely contained in $R$ as:
$\ 
  \pi =
  \pi_0 \conc \chi_1 \conc \pi_1 \conc \ldots \conc \chi_k \conc \pi_k\text{,} 
\ $
where for any $0\leq l \leq k$, the subpath $\chi_l$ is a crossing for $R$.
This last is called the \textit{redex crossing form} (RCF) of $\pi$ for $R$.
\\
The \textit{path reduction} is a function from paths in $\net$ to sums 
of paths in $\rho(\net)$.
The residual of $\pi$, written $\rho(\pi)$, is defined according to the 
reduction rule used by $\rho$ and by extension of the case of $\rho(\chi_l)$.
\begin{enumerate}
\item \label{def:RIN:path-red:imp}
  Linear implication cut.
  If $\chi_l$ is as in \autoref{eq:RIN:red:impl},
  then $\rho(\chi_l)$ is defined as follows.\\
  \begin{minipage}{0.45\textwidth}
    \begin{align}
      \label{eq:RIN:crossing-red:imp:ll}
      \rho((v,w,u)) &= (v)\\
    \label{eq:RIN:crossing-red:imp:rr}
    \rho((v',w,u') &= (v')
    \end{align}
  \end{minipage}
  \begin{minipage}{0.45\textwidth}
    \begin{align}
      \label{eq:RIN:crossing-red:imp:lr}
      \rho((v,w,u')) &= 0
      \\
      \label{eq:RIN:crossing-red:imp:rl}
      \rho((v',w,u)) &= 0
    \end{align}
  \end{minipage}\\\vspace{\lineskip}

  The residual of the whole $\pi$ is defined as:
  \begin{equation}
  \rho(\pi) =
  \begin{cases}
  \label{eq:RIN:path-red:imp}
  \pi_0 \conc \rho(\chi_1) \conc \pi_1 \conc \ldots
    \conc \rho(\chi_k) \conc \pi_k
    &\text{if for any $i$, } \rho(\chi_i)\neq0\\
  0
    &\text{otherwise}
  \end{cases}
  \end{equation}
\item \label{def:RIN:path-red:exp}
  Exponential cut.
  Let $\chi_l$ be as in \autoref{eq:RIN:red:exp} and $\permn \in \permns$.
  First, we define the residual of $\chi_l$ with respect to $\permn$,
  for every pair of indices $0 \leq i \leq n$, and $0 \leq j \leq m$:
  \begin{equation}
  \label{eq:RIN:crossing-red:exp}
  \rho^{\permn}(v_i,w,u_j) =
  \begin{cases}
  (v_j)
    &\text{if } n = m \text{, and }\permn(i) = j\\
  0
    &\text{if }n \neq m \text{, or }\permn(i) \neq j\\
  \end{cases}
  \end{equation}
  Now, similarly to \autoref{eq:RIN:path-red:imp}, we can define the residual 
  of the path $\pi$ with respect to $\permn$:
  \begin{equation}
    \label{eq:RIN:path-red-sigma:exp}
    \rho^{\permn}(\pi) =
    \begin{dcases}
      \pi_0 \conc
      \rho^{\permn}(\chi_1) \conc \pi_1 \conc \ldots \conc
      \rho^{\permn}(\chi_k) \conc \pi_k
    &\text{if for any } l \text{, }\rho^{\permn}(\chi_l)\neq0\\
    0
    &\text{otherwise}
    \end{dcases}
  \end{equation}
  Finally, we can define the residual of $\pi$ as the sum of all the residuals, 
  for any $\permn$:
  \begin{equation}
  \label{eq:RIN:path-red:exp}
    \rho(\pi) =
    \sum_{\permn \in \permns}
      \rho^{\permn}(\pi)
  \end{equation}
\end{enumerate}
If $\rho(\pi) \neq 0$, then $\pi$ is \textit{persistent to $\rho$}.
If, for every reduction sequence $\rho = (\rho_1, \ldots, \rho_m)$, and for 
every $1 \leq i \leq m$, the path $\pi$ is persistent to $\rho_i$, then $\pi$ 
is \textit{persistent}.
\end{definition}

\begin{example}\label{ex:path}
Recall the nets discussed in \autoref{ex:nets} and let
  $\rho(\net) = \net \reduce \net_l + \net_r$,
respectively be the left and the right addend of \autoref{fig:ex}.
Look at the net $\net $ and notice the paths $(v_1,v_4)$ and $(v_1,v_2,v_1)$ 
are not straight -- the former is twisting, while the latter is bouncing.
Consider the net $\transl{I}$ and the path $\phi = (w_1,w_2,w_3)$.
It is straight and also maximal.
Indeed, $\epaths{\transl{I}} = \lbrace \phi, \phi^- \rbrace$.
What about $\epaths{\net}$?
If we start from $v_1$ we find two paths seeking for the head variable:
  $\pi_1 = (v_1, v_2, v_3, z_1)$ and
  $\pi_2 = (v_1, v_2, v_3, z_2)$.
Both are straight and persistent, since:
  $\nf[\pi_1] = \pi_{1_r} = (v_1, v_2 \eqtd z_1)$ and
  $\nf[\pi_2] = \pi_{2_l} = (v_1, v_2 \eqtd z_2)$.
On the other hand they are not comprehensive, since they do not cross $v_4$ nor 
$v_7$.
Remark also $\pi_1,\pi_2$ cross the exponential redex differently,
and they do not belong to the same addend of the reduct,
for $\pi_{1_r} \in \net_r$, while $\pi_{2_l} \in \net_l$.
If otherwise we begin with $z_1$, which morally represents a free variable, 
the walk searches for the term that is going to substitute it.
$\pi_3 = (z_1, v_3, v_2, v_4, v_5, v_6, v_3, z_1)$
is not persistent, because crosses the same redex twice, each of those belongs 
to a distinct permutation.
Morally $\pi_3$ is trying to use the same variable as the function of both 
the applications.
\end{example}

\subsection{Results}

A persistent execution path travels through every vertex of a net that either 
belongs to the normal form of the net, or is eliminated by the normalisation.
Since RINs represent a linear calculus without erasing, we show that 
comprehensiveness of paths is a natural property for execution paths.
Moreover RINs have no duplication, despite what sum creation looks like,
so we can also show that the exponential reduction rule indeed partitions 
persistent execution paths among the addends it creates ---its action is a 
bijection.

\begin{lemma}
\label{lem:RIN:exponentials}
In a closed simple net $\net$,
the conclusion of an exponential link is either
the first premise of a linear implication link, or
a cut with another exponential link.
\end{lemma}
\begin{proof}
Given \autoref{def:RIN} of nets, we proceed by induction on the length of 
the reduction sequence $\rho: \transl{t} \reduce^* \sumnet$, 
for some term $t$ and net $\sumnet$.
\begin{enumerate}
\item \textit{Base.} Suppose $|\rho| = 0$.
  Thesis holds for $\transl{\cdot}$, by immediate verification of 
  \autoref{def:RIN:transl}.
\item \textit{Step.} Suppose $|\rho| > 0$.
  Let $l = \link{v_1, \ldots, v_n}{\wn/\oc}{v}$
    be an exponential link in $\net$ 
  and suppose $\rho = \rho' \rho''$,
    for some reduction sequence $\rho''$
      where the thesis holds by inductive hypothesis (IH),
        for some reduction step $\rho'$ that is of our interest.
  \begin{enumerate}
  \item If $\rho'$ does not affect $l$, then IH is trivially preserved.
  \item If $\rho'$ is a linear implication step involving $l$, then, by IH, the
    redex has to contain $\pimplink{v}{u}{w}$ and $\nimplink{v'}{u'}{w}$
    (if $l$ is negative), or 
    by $\pimplink{v'}{u'}{w}$ and $\nimplink{v}{u}{w}$
    (if $l$ is positive).
    In both cases, observe now that, by definition of nets, and in particular 
    by \autoref{def:prenet} and \ref{def:types},
    $v'$ must be the conclusion of an exponential link $l'$ dual to $l$.
    Therefore in the reduct of $\rho'$, $v\eqtd v'$ will be the conclusion 
    of $l'$ (i.e. an exponential cut).
  \item If $\rho'$ is an exponential implication step involving $l$, then 
    notice that, by IH, $v$ has to be the reduced cut.
    In such a case $l$ is erased,
    either with, or without, the whole net $\net$, 
    (depending on whether the arity mismatches or not),
    hence preserving the thesis.
  \end{enumerate}
\end{enumerate}
\vspace{-1.5em}
\end{proof}


\begin{lemma}
\label{lem:comprehensiveness}
For any term $t$ such that $\transl{t}:\const$,
any persistent path $\pi\in\epaths{\transl{t}}$ is comprehensive.
\end{lemma}
\begin{proof}
We shall prove a stronger thesis:
given a persistent path $\pi\in\epaths{\transl{t}}$,
a vertex $v \notin \pi$ if and only if
there exists a (co-)weakening $l$ such that $v \in C(l)$.
\begin{itemize}
\item
  The ``if'' direction of the thesis follows from a mere observation of the 
  \autoref{def:path} of execution paths.
  If $\pi$ includes a conclusion of a (co-)weakening,
  then $\pi$ is necessarily bouncing or non-maximal,
  in both cases contradicting the hypothesis that $\pi$ is an execution path.

\item 
  In order to prove the ``only if'' part of the thesis,
  let us first recall that, by \autoref{def:RIN},
  a simple net is either a translation of a term, or an addend in its reduct.
  We now go by induction on a sequence $\rho$ of expansion (or an 
  anti-sequence of reduction) from a normal form $\sumnet$
  back to a simple net $\net = \transl{t}$, for some term $t$.
  If $\sumnet = 0$ there is nothing to prove,
  so we shall assume it to be non-zero.

  \begin{enumerate}
  \item \textit{Base.}
    Suppose $|\rho|=0$.
    Then $\sumnet=\net=\transl{t}$.
    Therefore $t = \const$, because the only closed term whose translation is 
    normal with respect to net reduction is $\const$.
    Then $\net = \constlink{v}$, and $\epaths{\net} = \lbrace (v,v) \rbrace$.
  \item \textit{Step.}
    Suppose $|\rho| > 0$.
    Let $\rho: \net \reduce \sumnet'$ and $\sumnet' \reduce^* \sumnet$.
    We then distinguish two sub-cases depending on the rule employed by $\rho$.
    \begin{enumerate}
    \item \textit{Linear implication cut.}
      Suppose $\net'$ to be an addend of $\sumnet'$ containing the 
      vertices $v,u$, and the expansion step to be the following, where the 
      $v_1,v_2,u_1,u_2,w$ are introduced.
      \[\net'
        \quad\leftarrow\quad
        \net',\ \pimplink{u_1}{v_1}{w},\ \nimplink{u_2}{v_2}{w}\]
      Now $v$ cannot be the conclusion of a (co-)weakening,
      as established by \autoref{lem:RIN:exponentials}.
      Hence, by inductive hypothesis (IH), $v \in \pi$ and it is enough to 
      observe, by \autoref{def:RIN:path-red}, \autoref{def:RIN:path-red:imp},
      that also $v_1,v_2 \in \rho^{-1}(\pi)$.
      Let's now discuss $u$.
      \begin{enumerate}
      \item If $u \notin \pi$, then,
        because of \autoref{lem:RIN:exponentials}, 
        $u$ must be an exponential redex containing a weakening.
        So, first we clearly have $u_1 \notin \rho^{-1}(\pi)$.
        Moreover, we also have that $u_2 \notin \rho^{-1}(\pi)$ because 
        $u_2$ must be the conclusion of a co-weakening.
        Otherwise $u$ would be a net-neutralisation redex,
        which contradicts the persistence hypothesis we have for $\pi$.
        In such a case, the expansion is admissible, and we can verify both
        $u_1,u_2 \notin \rho^{-1}(\pi)$.
      \item Otherwise $u \in \pi$.
        Then again by inspection of
        \autoref{def:RIN:path-red} and \ref{def:RIN:path-red:imp},
        we verify that $u_1,u_2 \in \rho^{-1}(\pi)$.
      \end{enumerate}
    \item \textit{Exponential cut.}
      If the expansion affect $0$ addends,
      $\pi$ is unaffected, hence IH is trivially preserved.
      Otherwise, let $\rho$ be as follows.
      \begin{gather*}
        \sum_{\permn \in \permns}\ \net_C
        [v_1 \eqtd u_{\permn(1)}, \ldots, v_n \eqtd u_{\permn(n)}]\\
        \uparrow\\
        \net',\ \oclink{v_1,\ldots,v_n}{w},\ \wnlink{u_1,\ldots, u_n}{w}
      \end{gather*}
      For any $1\leq i,j \leq n$,
      the vertices $v_i,u_j$ cannot be the conclusion 
      of some (co-)weakening (cf. \autoref{lem:RIN:exponentials}).
      Thus, by IH, for any $1\leq i',j' \leq n$ such that
      $\rho(v_{i'}), \rho(u_{j'}) \in \pi$, we also have
      $v_{i'},u_{j'} \in \rho^{-1}(\pi)$.
    \end{enumerate}
  \end{enumerate}
\end{itemize}
\vspace{-1.5em}
\end{proof}

\begin{theorem}
\label{thm:RIN:path-red:bij}
For any closed $\transl{t}:\const$, every reduction step $\rho$ induces a 
bijection between persistent paths in $\epaths{\net}$ and persistent paths 
$\epaths{\rho(\net)}$.
\end{theorem}

\begin{proof}
Let $\pi \in \epaths{\net}$ be persistent, and suppose its RCF is
$\pi_0 \conc \chi_1 \conc \pi_1 \conc \ldots \conc \chi_k \conc \pi_k$.
There are two reduction rules possibly used by $\rho$.
\begin{enumerate}
\item
  \textit{Linear implication cut.}
  Because of the persistence of $\pi$ to $\rho$, and by the
  definition given by \autoref{eq:RIN:path-red:imp}, we have $\chi_l \neq 0$, 
  for all $0 \leq l \leq k$, and
    $\rho(\pi) =
    \pi_0 \conc \rho(\chi_1) \conc \pi_1 \conc \ldots
    \conc \rho(\chi_k) \conc \pi_k$.
  By an inspection of the definition of crossing path reduction, we first 
  notice that, for each of the four possible extrema of $\chi_l$, there is a 
  unique reduct that corresponds to a particular direction
  (from $\inpt$ to $\outpt$, or \textit{vice versa})
  of one of the two cut vertices in $\net'$: 
  $v$, and $v'$.
  Namely, let $\rho$ be as in \autoref{eq:RIN:red:impl}.
  Then, the bijection is given as follows:
  \begin{enumerate}
  \item
    $\chi_l = (v,w,u)$ if and only if
    $\rho(\chi_l) = v$ from $\inpt$ to $\outpt$;
  \item
    $\chi_l = (v',w,u')$ if and only if
    $\rho(\chi_l) = v'$ from $\outpt$ to $\inpt$;
  \item
    $\chi_l = (u,w,v)$ if and only if
    $\rho(\chi_l) = v$ from $\outpt$ to $\inpt$;
  \item
    $\chi_l = (u',w,v')$ if and only if
    $\rho(\chi_l) = v'$ from $\inpt$ to $\outpt$.
  \end{enumerate}
  Such a bijection holds between $\chi_l$ and $\rho(\chi_l)$, so we also
  have a bijection between $\pi$ and $\rho(\pi)$.
\item
  \textit{Exponential cut.}
  Suppose the redex $R$ being as in \autoref{eq:RIN:red:exp}.
  Because of the persistence of $\pi$ to $\rho$, and by the definition given by 
  \autoref{eq:RIN:path-red-sigma:exp} and \ref{eq:RIN:path-red:exp}, it must be 
  the case that $n=m$ and that there exist a permutation $\permn \in \permns$ 
  such that for all $0 \leq l \leq k$, we have $\rho^{\permn}(\chi_l) \neq 0$.
  Moreover, by \autoref{lem:comprehensiveness}, $\pi$ is 
  comprehensive, therefore 
  $\permn$ is unique and for any other $\permn'$, we have 
  $\rho^{\permn'}(\pi)=0$.
  Let $\chi_l$ be as in \autoref{eq:RIN:crossing-red:exp}, and observe it 
  must be also the case that $\permn(i)=j$ so that
  $\rho^{\permn}(v_i,w,u_j) = (v_{\permn(i)} \eqtd u_j)$.
  We then obtained a one-to-one relation between $\pi$ and $\rho(\pi)$.
\end{enumerate}
\vspace{-1.5em}
\end{proof}

\section{Execution} \label{sec:Ex}

\subsection{Definitions}
We are ready to formulate the GoI construction for RINs.
We followed the spirit of the formulation for the case of MELL as formulated in 
\cite{DanosRegnier:1995}, but we characterise our resource exponentials,
which have no promotion, as a sort of superposition of $n$-ary multiplicatives.
We define a weight assignment for paths, so that the execution of a net is
the sum of the weights of any execution path within it, and we formulate a 
Dynamic Algebra $\rlstar$ on weights representing the computation.
A crossing of an exponential link is weighted spanning over the space of 
permutations of the indices of link's premises, and exponential weights 
interacts exactly as multiplicatives weights, i.e. by nullification or 
neutralisation.

\begin{definition}[Dynamic Algebra]\label{def:rlstar}
The $\rlstar$ algebra is defined over symbols in
$\lbrace 0, 1, p, q, e_n, \wunit \rbrace$, where $n$ is a natural number.
A word of its alphabet, called \textit{weight}, is generated by an unary 
\textit{inversion} operator $(\cdot)^*$ and a binary \textit{concatenation} 
operator with infix implicit notation.
The concatenation operator is a monoid, whose identity element is $1$, and 
whose absorbing element is $0$ (cf. \autoref{fig:rlstar:monoid}).
Moreover, the inversion operator is idempotent and involutive for concatenation
(cf. \autoref{fig:rlstar:conc-inv}), and satisfies the neutralisation and 
two annihilation equations in \autoref{fig:rlstar:neutr-ann}.
\end{definition}

\begin{figure}
  \vspace{0.3em}
  \begin{subfigure}{0.28\textwidth}
    \caption{Monoid rules.}
    \label{fig:rlstar:monoid}
    \vspace{-1.5em}
    \begin{align}
    \label{eq:lstar:assoc}
      a(bc) =& (ab)c\\
    \label{eq:lstar:one}
      a1 = 1a &= a\\
    \label{eq:lstar:zero}
      a0 = 0a &= 0
    \end{align}
  \end{subfigure}
  \hfill
  \begin{subfigure}{0.23\textwidth}
    \caption{Inversion rules.}
    \label{fig:rlstar:conc-inv}
    \vspace{-1.5em}
    \begin{align}
    \label{eq:lstar:idemp}
      (a^*)^* &= a\\
    \label{eq:lstar:invol}
      (a b)^* &= b^* a^*\\
    \notag
    \end{align}
  \end{subfigure} 
  \hfill
  \begin{subfigure}{0.28\textwidth}
    \caption{Computation rules.}
    \label{fig:rlstar:neutr-ann}
    \vspace{-1.5em}
    \begin{align}
    \label{eq:lstar:neutr}
      a a^* &= 1\\
    \label{eq:lstar:annihil:impl}
      q p^* = p q^* &= 0\\
    \label{eq:rlstar:annihil:exp}
      e_i e_{j \neq i}^* &= 0
    \end{align}
  \end{subfigure}
  \caption{The $\rlstar$ algebra.}
\end{figure}

\begin{definition}[Weighting]\label{def:RIN:weight}
  The \textit{permuted base weighting} is a map $w$
  from a unitary straight path $\pi \in \net$
    and a resource permutation $\permnet$
  to a weight of $\rlstar$, written as $w^{\permnet}(\pi)$.
  \begin{equation}
  \label{eq:RIN:weight:unitary}
  w^{\permnet}((u,v)) = 
  \begin{dcases}
    \const 
      & \text{if there is } \constlink{u} \text{ and } u=v 
      \\
    p 
      & \text{if there is }
      \pimplink{u}{w}{v} \text{ or } 
      \nimplink{u}{w}{v}
      \\
    q 
      & \text{if there is }
      \pimplink{w}{u}{v} \text{ or } 
      \nimplink{w}{u}{v}
      \\
    e_i 
      & \text{if there is }
      \wnlink{u_1,\ldots,u_i,\ldots, u_n}{v}
      \text{ and } u_i = u
      \\
    e_{\perm_r(i)} 
      & \text{if there is }
      r = \oclink{u_1,\ldots,u_i,\ldots, u_n}{v}
      \text{ and } u_i = u
      \\
    (w^{\permnet}(v,u))^*
      & \text{otherwise.}
  \end{dcases}
  \end{equation}
  Straightness of the unitary path $\pi$ implies that $\pi$ goes either:
  (i) from a premise vertex to a conclusion one, i.e. crossing a link in one of 
  the five possible ways that are covered by the first five clauses;
  (ii) vice versa, from a conclusion vertex to a premise one, covered by the 
  last clause.\\
  The \textit{permuted weighting} is the lifting of the permuted base weighting 
  to generic straight paths, and the \textit{path weighting} is the sum of all 
  the permuted weights of a path, for any resource permutation:
  \begin{align}
  \label{eq:RIN:weight}
    \permweight{\permnet}{v} = 1
    &&
    \permweight{\permnet}{(u,v)\conc\pi} =
      w^{\permnet}(u,v) \permweight{\permnet}{\pi}
    &&
    \weight{\pi} = \sum_{\permnet\in\permnets} \permweight{\permnet}{\pi}
    \text{.}
  \end{align}
\end{definition}

\begin{definition}[Execution]
\label{def:RIN:path:compr}
A path $\pi$ is \textit{regular} if $\weight{\pi}\neq0$.
The \textit{execution} of a net $\sumnet$, is defined as:
\begin{equation}
  \exec{\net} = \sum_{\pi \in \ecpaths{\net}} \weight{\pi}\text{.}
\end{equation}
\end{definition}

\begin{example}\label{ex:exec}
Consider the again the closed net $\met$, whose reduction has been discusssed 
in previous \autoref{ex:nets}, and that is depicted in the leftmost extremity 
of 
\autoref{fig:ex2}.
To have an idea of the execution of $\met$ and of the behaviour of the 
algebra, let us consider an execution comprehensive path, one of the 
persistent two, and compute its weight.
Given that the path is palindromic, i.e. has the form $\pi \conc \pi^-$, we 
will consider only its first half, that goes from the root of the term to the 
constant.
Moreover we will break lines when a path invert its polarity direction, i.e. if 
it walks from $\inpt$ to $\outpt$ or viceversa.
\begin{align*}
  (&
    v_1,v_2,v_3,w_1,w_2,w_3,
  &\sum_{\sigma_2 \in S_2}\ &
    q\ e_1\ e^*_{\sigma(2)}\ q^*\ e_1 \cdot
  \\[-2ex]
  &
    w_1,v_3,v_2,
  &&
    p\ e_{\sigma(2)}\ e^*_1\cdot
  \\
  &
    v_4,v_5,v_6,v_3,z_1,z_2,z_3,
  &&
    p^*\ e^*_1\ q\ e_2\ e^*_{\sigma(1)}\ q^*\ e_1 \cdot
  \\
  &
    z_1,v_3,v_6,
  &&
    p\ e_{\sigma(1)}\ e^*_2\ \cdot
  \\
  &
    v_7,v_8)
  &&
    p^*\ e^*_1\ \star
  \\
\intertext{On the path: reduce it using the exponential rule.
  On the weight:
  apply \autoref{eq:rlstar:annihil:exp} and then \ref{eq:lstar:zero}
  on the addend s.t. $\sigma_2 = (2,1)$,
  apply \autoref{eq:lstar:neutr} and then \ref{eq:lstar:one}
  on the one s.t. $\sigma_2 = (1,2)$.}
  \reduce\ (&
    v_1,v_2 \eqtd w_1,w_2,w_3,
  &=_{\rlstar}\ &
    q\ q^*\ e_1 \cdot
  \\
  &
    w_1 \eqtd v_2,
  &&
    p\ \cdot
  \\
  &
    v_4,v_5,v_6 \eqtd z_1,z_2,z_3,
  &&
    p^*\ e^*_1\ q\ q^*\ e_1 \cdot
  \\
  &
    z_1 \eqtd v_6,
  &&
    p\ \cdot
  \\
  &
    v_7,v_8)\;+
  &&
    p^*\ e^*_1\ \star \;+
 \\
  &
    0
  &&
    0
\intertext{Forget zeros on both side.
  On the path: reduce it using the leftmost linear implication rule.
  On the weight: apply \autoref{eq:lstar:neutr} and then \ref{eq:lstar:one}.}
  \reduce\ (&
    v_1 \eqtd w_2,w_3 \eqtd v_4, v_5,v_6 \eqtd z_1,z_2,z_3,
  &=_{\rlstar}\ &
    e_1 e^*_1\ q\ q^*\ e_1 \cdot
  \\
  &
    z_1 \eqtd v_6,
  &&
    p\ \cdot
  \\
  &
    v_7,v_8)
  &&
    p^*\ e^*_1\ \star
 \\
\intertext{On the path: reduce it using the linear implication rule.
  On the weight: apply \autoref{eq:lstar:neutr} and then \ref{eq:lstar:one}.}
  \reduce\ (&
    v_1 \eqtd w_2,w_3 \eqtd v_4,v_5 \eqtd z_2,z_3 \eqtd v_7,v_8)
  &=_{\rlstar}\ &
    e_1\ e^*_1\ e_1\ e^*_1\ \star
\\
\intertext{On the path: reduce it twice using exponential rules.
  On the weight: apply \autoref{eq:lstar:neutr} and then \ref{eq:lstar:one}, 
  and repeat.}
  \reduce\ (&
    v_1 \eqtd w_2 \eqtd v_5 \eqtd z_2,z_3 \eqtd v_7,v_8)
  &=_{\rlstar}\ &
    e_1\ e^*_1\ \star
\\
  \reduce\ (&
    v_1 \eqtd w_2 \eqtd v_5 \eqtd z_2 \eqtd v_8)
  &=_{\rlstar}\ &
    \star
\end{align*}
Therefore the persistent path turns out to be regular.
Even more, along the reduction we managed to apply, for each step, some
$\rlstar$ equations so that weight of every reduct is equal to the manipulated 
weight.
The next two theorems shall generalise these two facts.
\end{example}

\subsection{Results}

The $\rlstar$ algebra introduced so far accurately computes path reduction.
We prove the equivalence between regularity and persistence, and
show execution is invariant by reduction.
Not only the GoI is a suitable semantic for ground typed RINs, but also possess 
quantitative-awareness, since, for any term, the cardinality of execution paths 
that are regular is equal to those of addends of its normal form.

\begin{lemma}\label{lem:RIN:weight-inv}
For any closed simple net $\net:\const$,
any reduction step $\rho$, and
any path $\pi \in \ecpaths{\net}$:
\begin{itemize}
  \item $\rho(\pi) \neq 0$, \  and \ 
    $\weight{\pi} =_\rlstar \weight{\rho(\pi)}$; or
  \item $\rho(\pi) = 0$, \ and \ 
    $\weight{\pi} =_\rlstar 0$.
\end{itemize}
\end{lemma}

\begin{proof}
Let $\pi \in \epaths{\net}$ and recall it has to be long enough for $\rho$, for 
it is maximal.
Suppose the RCF of $\pi$ w.r.t. the redex $R$ of $\rho$ is
  $\pi_0 \conc \chi_1 \conc \pi_1 \conc \ldots \conc \chi_k \conc \pi_k$.
We proceed by a case analysis of the reduction rule used by $\rho$.
\begin{enumerate}
\item \textit{Linear implication cut elimination.}
  Let $R$ be as in \autoref{eq:RIN:red:impl}.
  We distinguish two sub-cases, depending on the nullity of $\rho(\pi)$.

  \begin{enumerate}
  \item
    Suppose $\rho(\pi)=0$.
    Then by \autoref{def:RIN:path-red:imp} of \autoref{def:RIN:path-red},
    in particular \autoref{eq:RIN:path-red:imp},
    there must exist $0 \leq l \leq k$, such that
    $\rho(\chi_l) = 0$.
    Hence it must be the case that $\chi_l$ is either
    as in \autoref{eq:RIN:crossing-red:imp:lr}, or
    as in \autoref{eq:RIN:crossing-red:imp:rl}.

    \begin{enumerate}
     \item
      Suppose $\chi_l = (v,w,u')$.
      Then $\rho(\chi_l) = 0$,
      and $\weight{\chi_l} = p q^* =_\rlstar 0$.
    \item 
      Suppose $\chi_l = (v',w,u)$.
      Then $\rho(\chi_l) = 0$,
      and $\weight{\chi_l} = q p^* =_\rlstar 0$.
    \end{enumerate}
    By definition of weighting (Eq. \ref{eq:RIN:weight}) and using 
    the \autoref{eq:lstar:zero}, we conclude
    $\weight{\pi} =_\rlstar 0$.
    
  \item
    Suppose $\rho(\pi) \neq 0$.
    Then, again by definition given in \autoref{eq:RIN:path-red:imp}
    for any $0 \leq l \leq k$, $\rho(\chi_l) \neq 0$.
    Hence it must be the case that $\chi_l$ is either
    as in \autoref{eq:RIN:crossing-red:imp:ll}, or
    as in \ref{eq:RIN:crossing-red:imp:rr}.

    \begin{enumerate}
    \item
      Suppose $\chi_l = (v,w,u)$.
      Then $\rho(\chi_l) = (v' \eqtd u')$,
      and $\weight{\chi_l} = p p^* =_\rlstar 1$.
    \item
      Suppose $\chi_l = (v',w,u')$.
      Then $\rho(\chi_l) = (v' \eqtd u')$,
      and $\weight{\chi_l} = q q^* =_\rlstar 1$.
    \end{enumerate}
    Now, applying this fact on the definition given
    by \autoref{eq:RIN:path-red:imp}, and using \autoref{eq:lstar:one}
    we conclude.
    \begin{align*}
    \weight{\pi}
    &= \sum_{\permnet\in\permnets}
      \permweight{\permnet}{\pi_0}\ 
      \permweight{\permnet}{\chi_1}\ 
      \permweight{\permnet}{\pi_1}\ 
      \ldots
      \permweight{\permnet}{\chi_k}\ 
      \permweight{\permnet}{\pi_k}\ 
    \\
    &= \sum_{\permnet\in\permnets}
      \permweight{\permnet}{\pi_0}\ 
      \permweight{\permnet}{\pi_1}\ 
      \ldots
      \permweight{\permnet}{\pi_k}\ 
    \\
    &= \weight{\rho(\pi)}
    \end{align*}
  \end{enumerate}
 
\item \textit{Exponential cut elimination.}
  Let $R$ be as in \autoref{eq:RIN:red:exp}, and let $r$ be 
  the $\oc$-link involved in.
  We distinguish two sub-cases, depending on the nullity of $\rho(\pi)$.
  \begin{enumerate}
  \item Suppose $\rho(\pi) = 0$.
    Then by \autoref{def:RIN:path-red:exp} of \autoref{def:RIN:path-red},
    in particular \autoref{eq:RIN:path-red:exp},
    there are only two possible causes.
    \begin{enumerate}
    \item Arity mismatch,
      i.e. when $n \neq m$, where $n,m$ are the arities of the two links.
      Because of the hypothesis of comprehensiveness of $\pi$,
      it must be the case that $k \geq max(n,m)$.
      Then, whatever permutation $\permn \in \permns$ we choose for the 
      premises of the $\oc$-link in $R$,
      there always exists a crossing $\chi_l$, for some $0 \leq l \leq k$,
      such that $\chi_l = (u_{\permn(i)},w,v_j)$ and $\permn(i) \neq j$.
    \item Permutation incoherence,
      i.e. when $n = m$, but for any $\permn \in \permns$
      there exists a crossing $\chi_l \subseteq \pi$ such that
      $\permn(i) \neq j$.
      This morally happens when $\pi$ tries to use more than once a 
      resource, travelling from the same premise of the $\wn$-link to two 
      different premises of the corresponding $\oc$-link.
    \end{enumerate}
    Thus, in both cases there is a ``wrong'' crossing $\chi_l \subset \pi$ such 
    that, for any $\permnet$, we have
    $\permweight{\permnet}{\chi_l} = e_{{\permnet}_r(i)} e_j^*$ where
    $\permn(i) \neq j$.
    Hence, by \autoref{eq:rlstar:annihil:exp},
    $\permweight{\permnet}{\chi_l} = 0$.
    By definition of weighting (Eq. \ref{eq:RIN:weight}) and using 
    the \autoref{eq:lstar:zero}, we conclude
    $\weight{\pi} =_\rlstar 0$.

  \item Suppose $\rho(\pi) \neq 0$.
    Again by definition of path reduction, it must be the case that $n=m$,
    and that there exist a $\permn'' \in \permns$ such that
    for all $0 \leq l \leq k$ we have $\rho^{\permn'}(\chi_l) \neq 0$.
    In particular, let $\chi_l$ be as in \autoref{eq:RIN:crossing-red:exp},
    and observe it must be also the case that $\permn'(i)=j$,
    which allows $\rho^{\permn'}(v_i,w,u_j) = (v_{\permn(i)} \eqtd u_j)$.
    Moreover, by the comprehensiveness hypothesis for $\pi$,
    $\permn$ has to be unique, so that for any other $\permn''$, 
    we have $\rho^{\permn''}(\pi) = 0$.
    So, accordingly to this, let us split resource permutations $\permnets$
    into $\permnets' \cup \permnets''$,
    where the former is the set of those such that for any
    $\permnet' \in \permnets'$, $\permnet'(r) = \permn'$,
    while, symmetrically, the latter contains those such that for any 
    $\permnet'' \in \permnets''$, $\permnet'(r) \neq \permn'$.
    Hence, by \autoref{def:RIN:weight} we obtain
    \begin{align*}
     \weight{\pi}
      &= \sum_{\permnet'\in\permnets'} \permweight{\permnet'}{\pi}\ +\ 
        \sum_{\permnet''\in\permnets''} \permweight{\permnet''}{\pi}
      \\
      &=
        \sum_{\permnet'\in\permnets'}
          \permweight{\permnet'}{\pi_0}\ 
          \permweight{\permnet'}{\chi_1}\ 
          \permweight{\permnet'}{\pi_1}\ 
          \ldots\ 
          \permweight{\permnet'}{\chi_k}\ 
          \permweight{\permnet'}{\pi_k}\ 
        \ +\\
        &\phantom{=\ }
        \sum_{\permnet''\in\permnets''}
          \permweight{\permnet''}{\pi_0}\ 
          \permweight{\permnet''}{\chi_1}\ 
          \permweight{\permnet''}{\pi_1}\ 
          \ldots\ 
          \permweight{\permnet''}{\chi_k}\ 
          \permweight{\permnet''}{\pi_k}
      \text{.}
    \intertext{
      In the leftmost series, by \autoref{eq:lstar:neutr},
      $\permweight{\permnet'}{\chi_l} = e_{\permnet'(r)(i)} e_j^* = 1$.
      While in the rightmost, by \autoref{eq:rlstar:annihil:exp},
      $\permweight{\permnet''}{\chi_l} = e_{\permnet''(r)(i)} e_j^* = 0$,
      so it neutralises to $0$.
      Therefore we concluded:
    }
      &=_\rlstar
        \sum_{\permnet'\in\permnets'}
          \permweight{\permnet'}{\pi_0}\ 
          \permweight{\permnet'}{\pi_1}\ 
          \ldots\ 
          \permweight{\permnet'}{\pi_k}
      \\
      &= \weight{\rho(\pi)}
      \text{.}
    \end{align*}
  \end{enumerate}
\end{enumerate}
\vspace{-1.5em}
\end{proof}

\begin{lemma}\label{lem:RIN:weight-inv:seq}
For any closed simple net $\net:\const$,
any reduction sequence $\rho$, and
any path $\pi \in \ecpaths{\net}$:
\begin{itemize}
  \item $\rho(\pi) \neq 0$, \  and \ 
    $\weight{\pi} =_\rlstar \weight{\rho(\pi)}$; or
  \item $\rho(\pi) = 0$, \ and \ 
    $\weight{\pi} =_\rlstar 0$.
\end{itemize}
\end{lemma}

\begin{proof}
A straightforward induction on the length $n$ of the sequence $\rho$.
\begin{enumerate}
\item \textit{Base.}
  Suppose $n=0$.
  Trivially, $\rho(\pi) = \pi$, so $\weight{\rho(\pi)} = \weight{\pi}$.
\item \textit{Step.}
  Suppose $n>0$.
  Let $\rho = \rho' \rho''$, with 
    $\rho'$ a single step,
    $\rho''$ a sequence.
  \begin{enumerate}
  \item If both $\rho'(\pi) \neq 0$ and $\rho''(\rho'(\pi)) \neq 0$, then, by 
    \autoref{def:RIN:path-red} of path reduction, 
    $\rho''(\rho'(\pi)) = \rho(\pi)$.
    In this case, by previous \autoref{lem:RIN:weight-inv},
    $\weight{\rho'(\pi)} =_\rlstar \weight{\rho''(\rho'(\pi))}$.
    But, by inductive hypothesis we have that
     $ \weight{\pi} =_\rlstar \weight{\rho'(\pi)}$
    so we conclude.
  \item Otherwise $\rho'(\pi)$ or $\rho''(\rho'(\pi))$ are zero.
    \begin{enumerate}
    \item If $\rho'(\pi) = 0$, then by definition of path 
      reduction, $\rho(\pi) = \rho''(\rho'(\pi)) = 0$.
      But by inductive hypothesis, $\weight{\pi} =_\rlstar 0$, that is our 
      thesis.
    \item Otherwise $\rho'(\pi) \neq 0$ while $\rho''(\rho'(\pi)) = 0$.
      Then $\rho(\pi) = 0$ and, again by previous \autoref{lem:RIN:weight-inv}, 
      $\weight{\rho'(\pi)} =_\rlstar 0$.
      Hence, the thesis.
    \end{enumerate}
\end{enumerate}
\end{enumerate}
\end{proof}

\begin{theorem}
\label{thm:RIN:regular}
For any closed net $\sumnet: \const$, a path $\pi \in \ecpaths\sumnet$ is 
persistent if and only if $\pi$ is regular.
\end{theorem}
\begin{proof}
  Immediate from \autoref{lem:RIN:weight-inv:seq}.
\end{proof}

\begin{theorem}
\label{thm:RIN:weight-inv}
For any closed net $\sumnet: \const$ and any reduction sequence
$\rho$, $\exec{\sumnet} =_\rlstar \exec{\rho(\sumnet)}$.
\end{theorem}

\begin{proof}
  Consider a pair $\pi,\rho(\pi)$ of paths respectively belonging to 
  $\ecpaths{\sumnet}, \ecpaths{\rho(\sumnet)}$.
  Recall that, thanks to \autoref{thm:RIN:path-red:bij},
  there is a bijection between the two, so it does not matter which we choose 
  first.
  \begin{enumerate}
  \item
    If $\pi,\rho(\pi)$ are not persistent, there is nothing left to prove, 
    for their weights are both $0$, as proven
    by \autoref{lem:RIN:weight-inv:seq},
    and consequently they are neutral with respect to both
    $\exec{\sumnet}, \exec{\rho(\sumnet)}$.
  \item
    Otherwise $\pi,\rho(\pi)$ are persistent,
    so they both are regular.
    Moreover, once again by \autoref{lem:RIN:weight-inv:seq},
    $\weight{\pi} = \weight{\rho(\pi)}$.
    Hence the thesis.
  \end{enumerate}
\vspace{-1.5em}
\end{proof}

\begin{corollary}
\label{cor:RIN:path-red:number}
  For any term $\transl{t}:\const$, regular paths in 
  $\transl{t}$ are as many as (non-zero) addends in $\nf[t]$.
\end{corollary}
\begin{proof}
  By definition of the calculus and of its nets syntax,
  $\nf[t] = n \const$, for some natural number $n$.
  Clearly, $\ecpaths{\transl{\const}}$ contains a unique path, bouncing on the 
  unique vertex of $\transl{\const}$.
  Then $|\ecpaths{\transl{\nf[t]}}| = n$.
  But from last \autoref{thm:RIN:weight-inv},
  $\exec{\transl{t}} = \exec{\transl{\nf[t]}}$,
  therefore the thesis.
\end{proof}

\section*{Conclusion}

\paragraph{Summary}
We studied the notion of path persistence in a restriction of the Resource 
Calculus (RC) showing that, in spite of the non-determinism,
the reduction induces a bijection between paths.
We defined a proper Geometry of Interaction construction that:
characterises persistence by an algebra of weights, which are 
non-deterministically assigned to paths;
is invariant under reduction;
accurately counts addends of normal forms.
In the restricted setting where we are placed, the formulation is 
considerably simpler and stronger with respect to similar works.

\paragraph{Further research} 
Future investigations may easily extend the minimalist formulation from RC to 
a PCF-like resource calculus, where the restriction to ground types remains 
innocuous although allowing a real-programming-language-class expressivity.
Directions of ongoing investigation by the author includes 
the study of the connection between Taylor-Ehrhard expansion and GoI, 
exploiting the resource construction hereby presented.
This could offer a technique to represent approximations of 
infinite, but still meaningful, paths in a $\lambda$-term,
as in the spirit of B\"{o}hm trees.
Indeed, paths, expansion and B\"{o}hm trees, they all intimately share a 
particular strategy of computation, that is the head reduction.
Lastly, a deep study of paths in presence of both superposition 
and duplication, i.e. in the full differential $\lambda$-calculus,
is still missing.
In such a case, the shape of persistent crossings of an exponential redex 
does not necessarily respect the definition we gave by mean of fixed 
permutations, because different copies of the redex may want different 
resource assignments.
\bibliographystyle{eptcs}
\bibliography{Bibliography}
\end{document}